\documentclass{article}

\usepackage{microtype}
\usepackage{graphicx}
\usepackage{subfigure}
\usepackage{booktabs} 

\usepackage{hyperref}

\usepackage[accepted]{icml2024}

\usepackage{amsmath}
\usepackage{amssymb}
\usepackage{mathtools}
\usepackage{amsthm}

\usepackage[capitalize,noabbrev]{cleveref}

\graphicspath{ {./figures/} }
\usepackage{subcaption}
\usepackage{graphicx,wrapfig}
\usepackage{subfigure}
\usepackage{comment}
\usepackage{multirow}
\usepackage{adjustbox}

\theoremstyle{plain}
\newtheorem{theorem}{Theorem}[section]
\newtheorem{proposition}[theorem]{Proposition}

\theoremstyle{definition}
\newtheorem{definition}[theorem]{Definition}

\theoremstyle{remark}

\usepackage[textsize=tiny]{todonotes}

\icmltitlerunning{LAGMA: LAtent Goal-guided Multi-Agent Reinforcement Learning}

\begin{document}

\twocolumn[
\icmltitle{LAGMA: LAtent Goal-guided Multi-Agent Reinforcement Learning}



\icmlsetsymbol{equal}{*}

\begin{icmlauthorlist}
\icmlauthor{Hyungho Na}{yyy}
\icmlauthor{Il-Chul Moon}{yyy,comp}
\end{icmlauthorlist}

\icmlaffiliation{yyy}{Korea Advanced Institute of Science and Technology (KAIST), Daejeon 34141, Republic of Korea.}
\icmlaffiliation{comp}{summary.ai, Daejeon, Republic of Korea}

\icmlcorrespondingauthor{Hyungho Na}{gudgh723@gmail.com}
\icmlcorrespondingauthor{Il-Chul Moon}{icmoon@kaist.ac.kr}

\icmlkeywords{Machine Learning, ICML}

\vskip 0.3in
]



\printAffiliationsAndNotice{}  

\begin{abstract}
In cooperative multi-agent reinforcement learning (MARL), agents collaborate to achieve common goals, such as defeating enemies and scoring a goal.
However, learning goal-reaching paths toward such a semantic goal takes a considerable amount of time in complex tasks and the trained model often fails to find such paths. 
To address this, we present \textbf{LA}tent \textbf{G}oal-guided \textbf{M}ulti-\textbf{A}gent reinforcement learning (\textbf{LAGMA}), which generates a goal-reaching trajectory in latent space and provides a latent goal-guided incentive to transitions toward this reference trajectory. 
LAGMA consists of three major components: (a) quantized latent space constructed via a modified VQ-VAE for efficient sample utilization, (b) goal-reaching trajectory generation via extended VQ codebook, and (c) latent goal-guided intrinsic reward generation to encourage transitions towards the sampled goal-reaching path. 
The proposed method is evaluated by StarCraft II with both dense and sparse reward settings and Google Research Football. Empirical results show further performance improvement over state-of-the-art baselines.

\end{abstract}
\section{Introduction}
\label{sec:introduction}


Centralized training and decentralized execution (CTDE) paradigm \cite{oliehoek2008optimal,gupta2017cooperative} especially with value factorization framework \citep{sunehag2017value,rashid2018qmix,wang2020qplex} has shown its success on various cooperative multi-agent tasks \cite{lowe2017multi,samvelyan2019starcraft}. However, in more complex tasks with dense reward settings, such as super hard maps in  StarCraft II Multi-agent Challenge (SMAC) \cite{samvelyan2019starcraft} or in sparse reward settings, as well as Google Research Football (GRF) \cite{kurach2020google}; learning optimal policy takes long time, and trained models even fail to achieve a common goal, such as destroying all enemies in SMAC or scoring a goal in GRF. Thus, researchers focus on sample efficiency to expedite training \cite{zheng2021episodic} and encourage committed exploration \cite{mahajan2019maven,yang2019hierarchical,wang2019influence}.

To enhance sample efficiency during training, state space abstraction has been introduced in both model-based \cite{jiang2015abstraction,zhu2021deepmemory,hafner2020mastering} and model-free settings \cite{grzes2008multigrid,tang2020discretizing,li2023neural}. Such sample efficiency can be more important in sparse reward settings since trajectories in a replay buffer rarely experience positive reward signals. However, such methods have been studied within a single-agent task without expanding to multi-agent settings.

To encourage committed exploration, goal-conditioned reinforcement learning (GCRL) \cite{kaelbling1993learning,schaul2015universal,andrychowicz2017hindsight} has been widely adopted in a single agent task, such as complex path finding with a sparse reward \cite{nasiriany2019planning,zhang2020generating,chane2021goal,kim2023imitating,lee2023cqm}. However, GCRL concept has also been limitedly applied to multi-agent reinforcement learning (MARL) tasks since there are various difficulties: 1) a goal is not explicitly known, only a semantic goal can be found during training by reward signal; 2) partial observability and decentralized execution in MARL makes impossible to utilize path planning with global information during execution, only allowing such planning during centralized training; 3) most MARL tasks seek not the shortest path, but the coordinated trajectory, which renders single-agent path planning in GCRL be too simplistic in MARL tasks. 

Motivated by methods employed in single-agent tasks, we consider a general cooperative MARL problem as finding trajectories toward semantic goals in latent space. 

\textbf{Contribution.} This paper presents \textbf{LA}tent \textbf{G}oal-guided \textbf{M}ulti-\textbf{A}gent reinforcement learning (\textbf{LAGMA}). LAGMA generates a goal-reaching trajectory in latent space and provides a latent goal-guided incentive to transition toward this reference trajectory during centralized training.
\begin{itemize}
\vspace{-0.1in}
\item \textbf{Modified VQ-VAE for quantized embedding space construction:} 
As one measure of efficient sample utilization, we use Vector Quantized-Variational Autoencoder(VQ-VAE) \cite{van2017neural} which projects states to a quantized vector space so that a common latent can be used as a representative for a wide range of embedding space. However, state distributions in high dimensional MARL tasks are quite limited to small feasible subspace unlike image generation tasks, whose inputs or states often utilize a full state space. In such a case, only a few quantized vectors are utilized throughout training when adopting the original VQ-VAE. To make quantized embedding vectors distributed properly over the embedding space of feasible states, we propose a modified learning framework for VQ-VAE with a novel \textit{coverage loss}. 
\vspace{-0.1in}
\item \textbf{Goal-reaching trajectory generation with extended VQ codebook:} 
LAGMA constructs an extended VQ codebook to evaluate the states projected to a certain quantized vector and generate a goal-reaching trajectory based on this evaluation. Specifically, during training, we store various goal-reaching trajectories in a quantized latent space. Then, LAGMA uses them as a reference to follow during centralized training. 
\vspace{-0.1in}
\item \textbf{Latent goal-guided intrinsic reward generation:} 
To encourage coordinated exploration toward reference trajectories sampled from the extended VQ codebook, LAGMA presents a latent goal-guided intrinsic reward. The proposed latent goal-guided intrinsic reward aims to accurately estimate TD-target for transitions toward goal-reaching paths, and we provide both theoretical and empirical support.
\end{itemize}

\section{Related Works}
\label{sec:related_works}
\textbf{State space abstraction for RL}\;
State abstraction groups states with similar characteristics into a single cluster, and it has been effective in both model-based RL \cite{jiang2015abstraction,zhu2021deepmemory,hafner2020mastering} and model-free settings \cite{grzes2008multigrid,tang2020discretizing}.
NECSA \cite{li2023neural} adopted the abstraction of grid-based state-action pair for episodic control and achieved state-of-the-art (SOTA) performance in a general single-RL task. This approach could relax the limitations of inefficient memory usage in the conventional episodic control, but this requires an additional dimensionality reduction technique, such as random projection \cite{dasgupta2013experiments} in high-dimensional tasks. 
Recently, EMU \cite{na2024efficient} presented a semantic embedding for efficient memory utilization, but it still resorts to the episodic buffer, which requires storing both the states and the embeddings. This additional memory usage could be burdensome in tasks with large state space.
In contrast to previous research, we employ VQ-VAE for state embedding and estimate the overall value of abstracted states. 
In this manner, a sparse or delayed reward signal can be utilized by a broad range of states, particularly those in proximity. In addition, thanks to the discretized embeddings, the count-based estimation can be adopted to estimate the value of states projected to each discretized embedding.
Then, we generate a reference or goal-reaching trajectory based on this evaluation in quantized vector space and provide an incentive for transitions that overlap with this reference.

\textbf{Intrinsic incentive in RL}\;
In reinforcement learning, balancing exploration and exploitation during training is a paramount issue \cite{sutton2018reinforcement}. To encourage a proper exploration, researchers have presented various forms of methods in a single-agent case such as modified count-based methods \cite{bellemare2016unifying, ostrovski2017count, tang2017exploration}, prediction error-based methods \cite{stadie2015incentivizing, pathak2017curiosity, burda2018exploration, kim2018emi}, and information gain-based methods \cite{mohamed2015variational, houthooft2016vime}. In most cases, an incentive for exploration is introduced as an additional reward to a TD target in Q-learning or a regularizer to overall loss functions. 
Recently, diverse approaches mentioned earlier have been adopted in the multi-agent environment to promote exploration \cite{mahajan2019maven, wang2019influence, jaques2019social, mguni2021ligs}.
As an example, EMC \cite{zheng2021episodic} utilizes episodic control \cite{lengyel2007hippocampal,blundell2016model} as regularization for the joint Q-learning, in addition to a curiosity-driven exploration by predicting individual Q-values. 
Learning with intrinsic rewards becomes more important in sparse reward settings. However, this intrinsic reward can adversely affect the overall policy learning if it is not properly annealed throughout the training. Instead of generating an additional reward signal solely encouraging exploration, LAGMA generates an intrinsic reward that guarantees a more accurate TD-target for Q-learning, yielding additional incentive toward a goal-reaching path.

    \begin{figure*}[!hbt]
    \begin{center}
    \centerline{\includegraphics[width=0.9\linewidth]{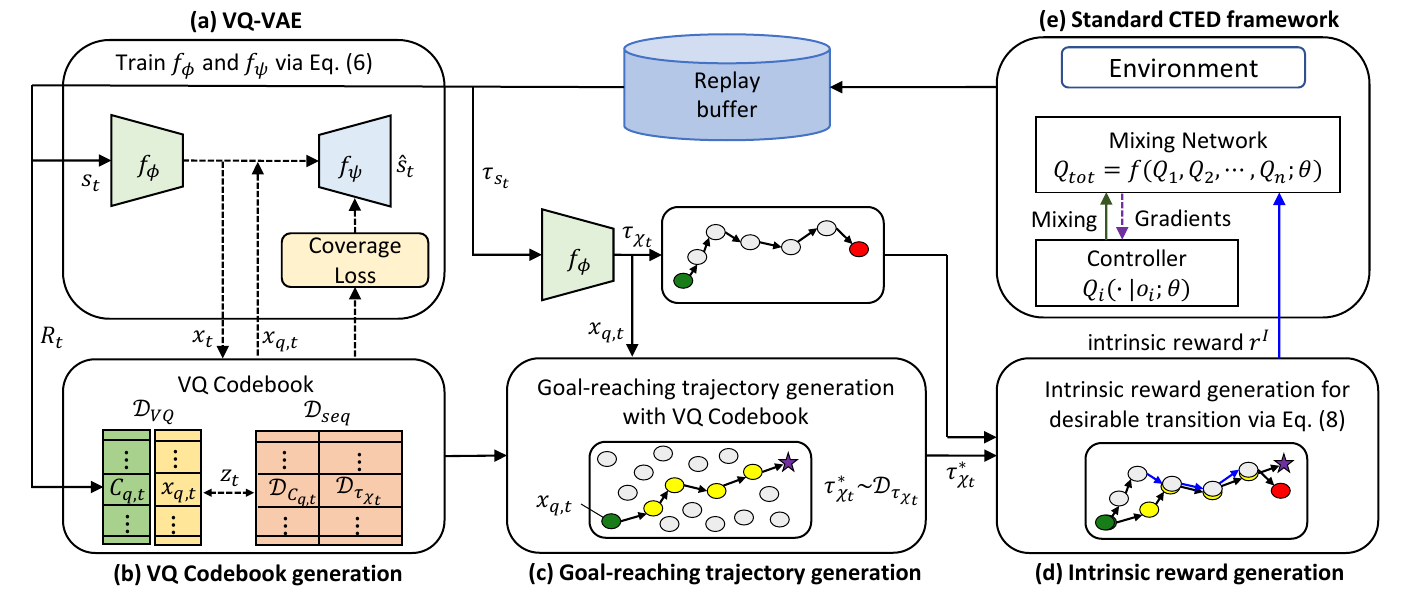}}
    \caption{Overview of LAGMA framework.(a) VQ-VAE constructs quantized vector space with coverage loss, while (b) VQ codebook stores goal-reaching sequences from a given $x_{q,t}$. Then, (c) the goal-reaching trajectory is compared with the current batch trajectory to generate (d) intrinsic reward. MARL training is done by (e) the standard CTDE framework.}
    \label{fig:overall_lagma}
    \end{center}
    \vspace{-0.30in}
    \end{figure*}
    
Additional related works regarding goal-conditioned reinforcement learning (GCRL) and subtask-conditioned MARL are presented in Appendix \ref{app:additional_related}.

\section{Preliminaries}
\label{sec:preliminaries}
\textbf{Decentralized POMDP} \;
A general cooperative multi-agent task with $n$ agents can be formalized as the Decentralized Partially Observable Markov Decision Process (Dec-POMDP) \cite{oliehoek2016concise}. DecPOMDP consists of a tuple $G = \left\langle {I,S,A,P,R,\Omega,O,n,\gamma} \right\rangle $, where $I$ is the finite set of $n$ agents; $s \in S$ is the true state in the global state space $S$; $A$ is the action space of each agent's action $a_i$ forming the joint action $\boldsymbol{a} \in {A^n}$; $P(s'|s,\boldsymbol{a})$ is the state transition function determined by the environment; $R$ is a reward function $r = R(s,\boldsymbol{a},s') \in \mathbb{R}$; $O$ is the observation function generating an individual observation from observation space $\Omega$, i.e., $o_i \in \Omega$; and finally, $\gamma \in [0,1)$ is a discount factor. In a general cooperative MARL task, an agent acquires its local observation $o_i$ at each timestep, and the agent selects an action ${a_i} \in A$ based on $o_i$. 
$P(s'|s,\boldsymbol{a})$ determines a next state $s'$ for a given current state $s$ and the joint action taken by agents $\boldsymbol{a}$. For a given tuple of $\left\{s, \boldsymbol{a}, s' \right\}$, $R$ provides an identical common reward to all agents. 
To overcome the partial observability in DecPOMDP, each agent often utilizes a local action-observation history ${\tau_i} \in T \equiv (\Omega \times A)$ for its policy ${\pi_i}(a|{\tau_i})$, where ${\pi}:T \times {A} \to [0,1]$ \cite{hausknecht2015deep,rashid2018qmix}. 
Additionally, we denote a group trajectory as $\boldsymbol{\tau}=<\tau_1,...,\tau_n>$.

\textbf{Centralized Training with Decentralized Execution (CTDE)} \;
In fully cooperative MARL tasks, under the CTDE paradigm, value factorization approaches have been introduced by \cite{sunehag2017value,rashid2018qmix,son2019qtran,rashid2020weighted,wang2020qplex} and achieved state-of-the-art performance in complex multi-agent tasks such as SMAC \cite{samvelyan2019starcraft}. 
In value factorization approaches, the joint action-value function $Q_{\theta}^{tot}$ parameterized by $\theta$ is trained to minimize the following loss function.
\begin{equation}\label{Eq:CTDE_loss}
\begin{aligned}    \mathcal{L}({\theta})=\mathbb{E}_{\boldsymbol{\tau},\boldsymbol{a},r^{\textrm{ext}},\boldsymbol{\tau}'\in \mathcal{D}}[\left( r^{\textrm{ext}} + \gamma V_{\theta^-}^{tot}(\boldsymbol{\tau}') - {Q_{\theta}^{tot}}(\boldsymbol{\tau},{\boldsymbol{a}}) \right)^2]
\end{aligned}
\end{equation}
Here, $V_{\theta^-}^{tot}(\boldsymbol{\tau}')={{\max }_{{\boldsymbol{a}'}}}{Q_{\theta^{-}}^{tot}}({\boldsymbol{\tau}'},{\boldsymbol{a}'})$ by definition; $\mathcal{D}$ represents the replay buffer; $r^{\textrm{ext}}$ is an external reward provided by the environment; $Q_{\theta^{-}}^{tot}$ is a target network parameterized by $\theta^-$ for double Q-learning\cite{hasselt2010double,van2016deep}; and $Q_{\theta}^{tot}$ and $Q_{\theta^{-}}^{tot}$ include both mixer and individual policy network. 

\textbf{Goal State and Goal-Reaching Trajectory} \;
In general cooperative multi-agent tasks, undiscounted reward sum, i.e., $R_0=\Sigma_{t=0}^{T-1}r_t$, is maximized as $R_{\textrm{max}}$ if agents achieve a semantic goal, such as defeating all enemies in SMAC or scoring a goal in GRF. 
Thus, we define goal states and the goal-reaching trajectory in cooperative MARL as follows.
\begin{definition}
\label{def:goal_reaching_traj}
(Goal State and Goal-Reaching Trajectory) For a given task dependent $R_{\textrm{max}}$ and an episodic sequence $\mathcal{T}:=\left\{s_0,\boldsymbol{a_0},r_0,s_1,\boldsymbol{a_1},r_1,...,s_T\right\}$, when $\Sigma_{t=0}^{T-1}r_t=R_{\textrm{max}}$ for $r_t\in\mathcal{T}$, we define such an episodic sequence as a goal-reaching sequence and denote as $\mathcal{T}^*$. Then, for $\forall s_t \in \mathcal{T}^*$, $\tau_{s_t}^*:=\left\{s_t,s_{t+1},...s_{T}\right\}$ is a goal-reaching trajectory and we define the final state of $\tau_{s_t}^*$ as a goal state denoted by $s_T^*$.
\end{definition}

\section{Methodology}
\label{sec:methodology}
This section introduces \textbf{LA}tent \textbf{G}oal-guided \textbf{M}ulti-\textbf{A}gent reinforcement learning (LAGMA) (Figure \ref{fig:overall_lagma}). We first explain how to construct a proper \textbf{(1) quantized embeddings via VQ-VAE}. To this end, we introduce a novel loss term called \textit{coverage loss} to distribute quantized embedding vectors across the overall embedding space. Then, we elaborate on the details of \textbf{(2) goal-reaching trajectory generation} with extended VQ codebook. Finally, we propose \textbf{(3) a latent goal-guided intrinsic reward} which guarantees a better TD-target for policy learning and thus yields a better convergence on optimal policy.
\begin{figure*}[!htb]
    \centering
    \subfigure[Training without $\mathcal{L}_\textrm{cvr}$ ($\lambda_{\textrm{cvr}}=0.0$).]{\includegraphics[width=0.33\linewidth]{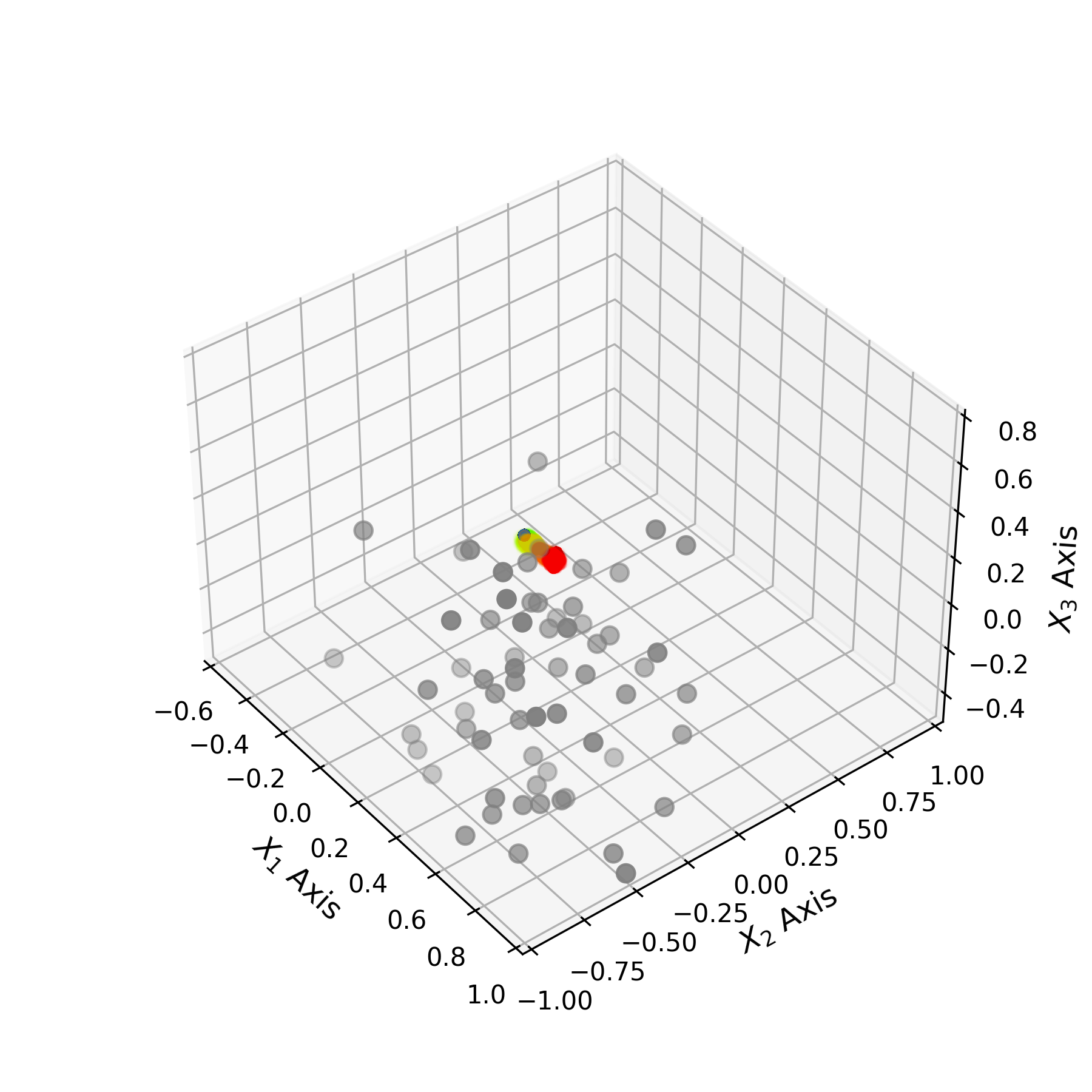}}    
    \subfigure[Training with ${\mathcal{L}_\textrm{cvr}^{\textrm{all}}}$ ($\lambda_{\textrm{cvr}}=0.2$).]{\includegraphics[width=0.33\linewidth]{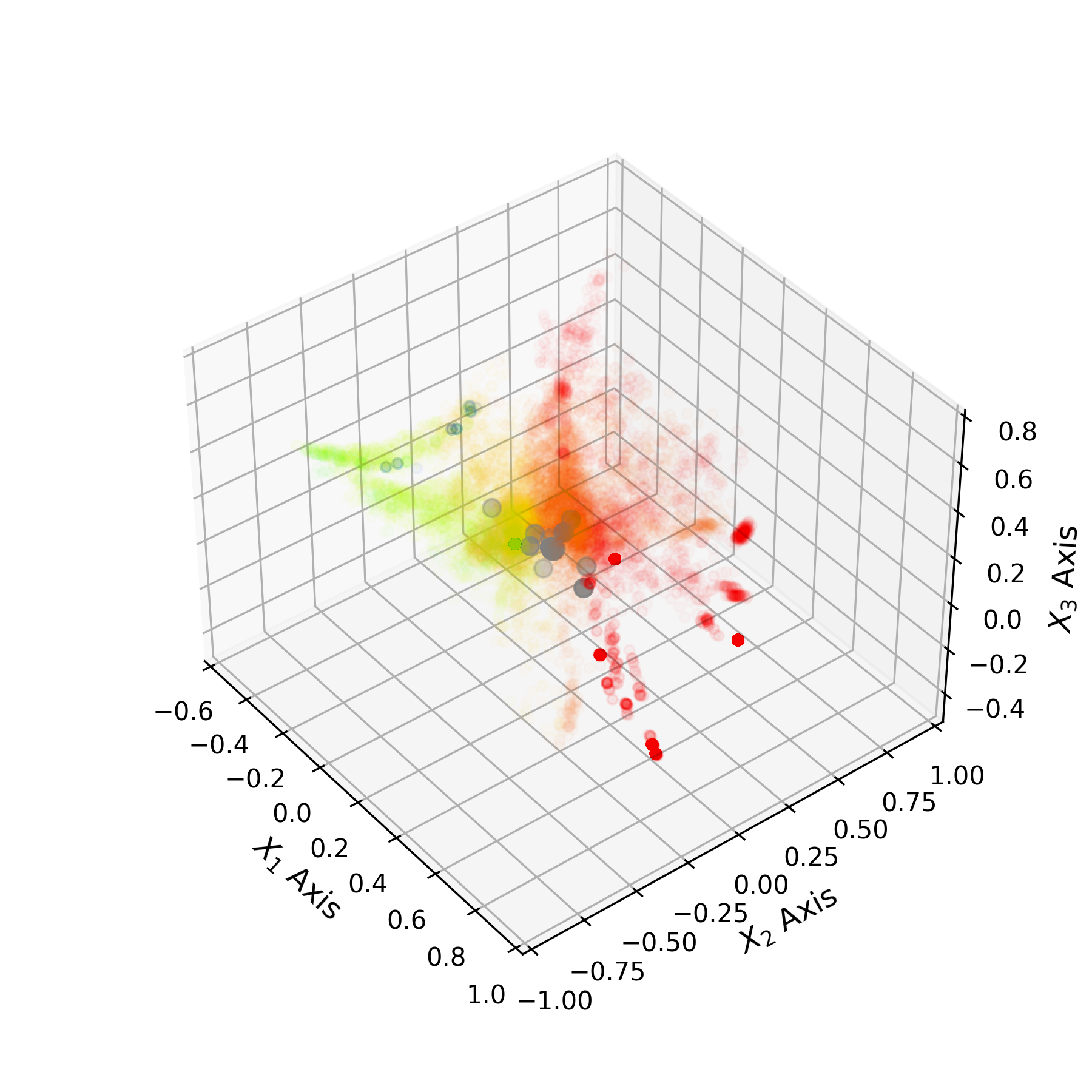}}
    \subfigure[Training with ${\mathcal{L}_\textrm{cvr}}$ ($\lambda_{\textrm{cvr}}=0.2$).]{\includegraphics[width=0.33\linewidth]{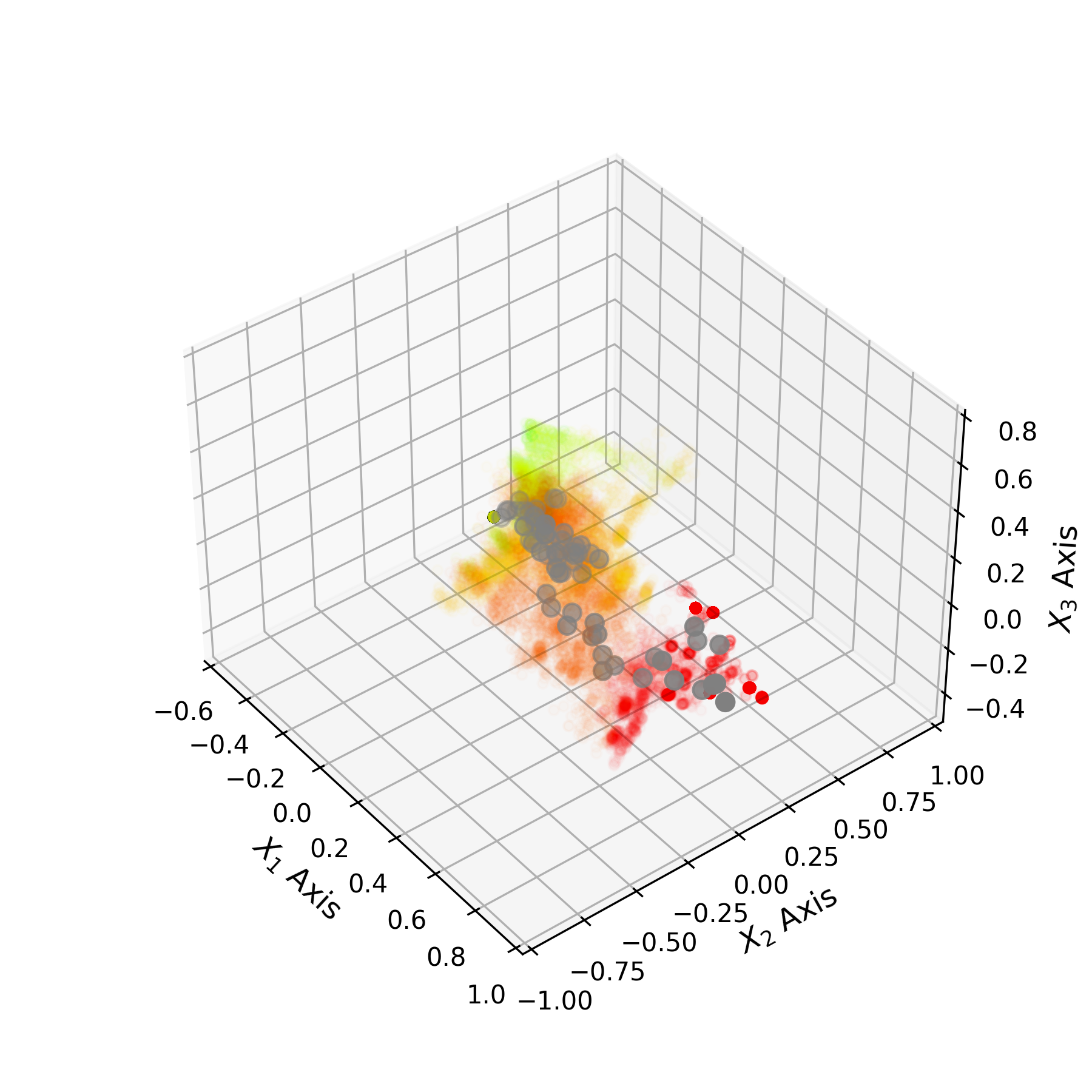}}       
    \vskip -0.10in	
    \caption{Visualization of embedding results via VQ-VAE. Under SMAC \texttt{5m\_vs\_6m} task, the size of codebook $n_c=64$, the latent dimension $D=8$; this illustrates embeddings at training time at T=1.0M. Colored dots represent $\chi$, which is a state presentation before quantization, and gray dots are quantized vector representations belonging to VQ codebook derived from the state representations. Colors from red to purple (rainbow) represent from small to large timestep within episodes.}
    \label{fig:vqvae_embedding}
    \vskip -0.15in
\end{figure*}
\subsection{State Embedding via Modified VQ-VAE}
\label{subsec:Latent_embedding}
In this paper, we adopt VQ-VAE as a discretization bottleneck \cite{van2017neural} to construct a discretized low-dimensional embedding space. Thus, we first define $n_c$-trainable embedding vectors (codes) $e_j \in \mathbb{R}^D$ in the codebook where $j=\{1,2,...,n_c\}$. An encoder network $f_{\phi}$ in VQ-VAE projects a global state $s$ toward $D$-dimensional vector, $x=f_{\phi}(s) \in \mathbb{R}^D$. Instead of a direct usage of latent vector $x$, we use a discretized latent $x_q$ by the \textit{quantization process} which maps an embedding vector $x$ to the nearest embedding vector in the codebook as follows.
\begin{equation}\label{Eq:quantization}
	x_q=e_{z}, \textrm{where}\; z=\textrm{argmin}_j||x-e_j||_2
\end{equation}
Then, the quantized vector $x_q$ is used as an input to a decoder $f_{\psi}$ which reconstructs the original state $s$. To train an encoder, a decoder, and embedding vectors in the codebook, we consider the following objective similar to \cite{van2017neural,islam2022discrete,lee2023cqm}. 
\begin{equation}\label{Eq:VQ_loss}
\begin{aligned}
	&{\mathcal{L}_{VQ}}(\phi,\psi,\boldsymbol{e}) = ||f_\psi([x=f_{\phi}(s)]_q) - s||_2^2\\ &+ {\lambda _{{\rm{vq}}}}||{\rm{sg[}}{f_\phi }(s)] - {x_q}||_2^2 
     + {\lambda _{{\rm{commit}}}}||{f_\phi }(s) - {\rm{sg[}}{x_q}]||_2^2
\end{aligned}
\end{equation}
Here, $[\cdot]_q$ and $\textrm{sg}[\cdot]$ represent a quantization process and stop gradient, respecitvely. $\lambda _{{\rm{vq}}}$ and $\lambda _{{\rm{commit}}}$ are scale factor for correponsding terms. The first term in Eq. \eqref{Eq:quantization} is the reconstruction loss, while the second term represents VQ-objective which makes an embedding vector $e$ move toward $x=f_{\phi}(s)$. The last term called a commitment loss enforces an encoder to generate $f_{\phi}(s)$ similar to $x_q$ and prevents its output from growing significantly. To approximate the gradient signal for an encoder, we adopt a straight-through estimator \cite{bengio2013estimating}.

When adopting VQ-VAE for state embedding, we found that only a few quantized vectors $e$ in the codebook are selected throughout an episode, which makes it hard to utilize such a method for meaningful state embedding. 
We presumed that the reason is the narrow projected embedding space from feasible states compared to a whole embedding space, i.e., $\mathbb{R}^D$. Thus, most randomly initialized quantized vectors $e$ locate far from the latent space of states in the current replay buffer $\mathcal{D}$, denoted as $\chi=\{x\in \mathbb{R}^D:x=f_{\phi}(s),s\in \mathcal{D}\}$, leaving only a few $e$ close to $x$ within an episode.
To resolve this issue, we introduce the \textit{coverage loss} which minimizes the overall distance between the current embedding $x$ and all vectors in the codebook, i.e., $e_j$ for all $j=\{1,2,...,n_c\}$. 
\vspace{-0.125in}
\begin{equation}\label{Eq:coverage_loss_all}
\begin{aligned}
	{\mathcal{L}_\textrm{cvr}^{\textrm{all}}}(\boldsymbol{e}) = \frac{1}{{{n_c}}}\sum\limits_{j = 1}^{{n_c}} {||{\rm{ }}} {\rm{sg[}}{f_\phi }(s)] - e_j||_2^2
\end{aligned}
\end{equation}
Although ${\mathcal{L}_\textrm{cvr}^{\textrm{all}}}$ could lead embedding vectors toward $\chi$, all quantized vectors tend to locate the center of $\chi$ rather than densely covering whole $\chi$ space. Thus, we consider a \textit{timestep dependent indexing} $\mathcal{J}(t)$ when computing the coverage loss. The purpose of introducing $\mathcal{J}(t)$ is to make only sequentially selected quantized vectors close to the current embedding $x_t$ so that quantized embeddings are uniformly distributed across $\chi$ according to timesteps. Then, the final form of coverage loss can be expressed as follows.
\begin{equation}\label{Eq:coverage_loss}
\begin{aligned}
	{\mathcal{L}_\textrm{cvr}}(\boldsymbol{e}) = \frac{1}{{{|\mathcal{J}(t)|}}}\sum\limits_{j\in\mathcal{J}(t)} {||{\rm{ }}} {\rm{sg[}}{f_\phi }(s)] - e_j||_2^2
\end{aligned}
\end{equation}
We defer the details of $\mathcal{J}(t)$ construction to Appendix \ref{app:details_implementation}. By considering the coverage loss in Eq. \eqref{Eq:coverage_loss}, not only the nearest quantized vector but also all vectors in the codebook move towards overall latent space $\chi$. In this way, $\chi$ can be well covered by quantized vectors in the codebook. Thus, we consider the overall learning objective as follows.
\begin{equation}\label{Eq:overall_obj_vqvae}
\begin{aligned}
	{\mathcal{L}_{VQ}^{tot}}(\phi,\psi,\boldsymbol{e}) = {\mathcal{L}_{VQ}}(\phi,\psi,\boldsymbol{e}) + \lambda _{{\rm{cvr}}} {\mathcal{L}_\textrm{cvr}}(\boldsymbol{e})
\end{aligned}
\end{equation}
where $\lambda _{{\rm{cov}}}$ is a scale factor for ${\mathcal{L}_\textrm{cvr}}$.

Figure \ref{fig:vqvae_embedding} presents the visualization of embeddings by principal component analysis (PCA) \cite{wold1987principal}. In Figure \ref{fig:vqvae_embedding}, the training without ${\mathcal{L}_\textrm{cvr}}$ leads to quantized vectors that are distant from $\chi$. 

In addition, embedding space $\chi$ itself distributes around a few quantized vectors due to the commitment loss in Eq. \eqref{Eq:VQ_loss}. Considering ${\mathcal{L}_\textrm{cvr}^{\textrm{all}}}$ makes quantized vectors close to $\chi$ but they majorly locate around the center of $\chi$ rather than distributed properly. 
\begin{wrapfigure}{l}{4.0cm}
	\includegraphics[width=4.0cm]{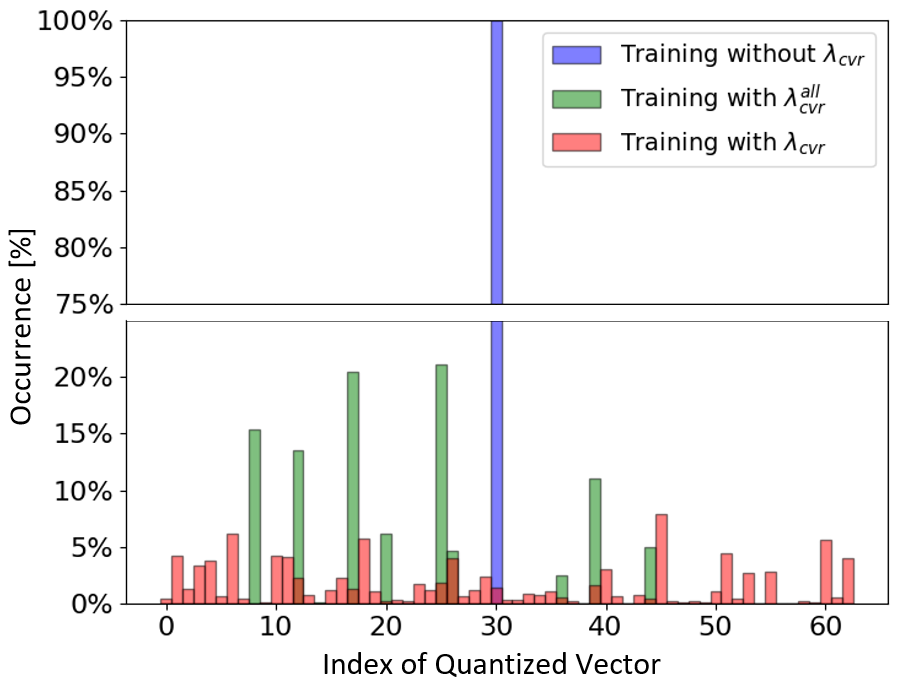}
	\vspace {-0.23in}
	\caption{Histogram of recalled quantized vector.}
	\label{fig:coordinated_trj}
\end{wrapfigure}
On the other hand, the proposed ${\mathcal{L}_\textrm{cvr}}$ results in well-distributed quantized vectors over $\chi$ space so that they can properly represent latent space of $s\in \mathcal{D}$. Figure \ref{fig:coordinated_trj} presents the occurrence of recalled quantized vectors for state embeddings in Fig. \ref{fig:vqvae_embedding}. We can see that training with $\lambda_{\textrm{cvr}}$ guarantees quantized vectors well distributed across $\chi$. Appendix \ref{app:details_implementation} presents the training algorithm for the proposed VQ-VAE.

\subsection{Goal-Reaching Trajectory Generation with Extended VQ Codebook}
\label{subsec:trajectory_generation}
After constructing quantized vectors in the codebook, we need to properly estimate the value of states projected to each quantized vector. Note that the estimated value of each quantized vector is used when generating an additional incentive to desired transitions, i.e., transition toward a goal-reaching trajectory.
Thanks to the quantized vectors in the codebook, we can resort to count-based estimation for the value estimation of a given state. For a given $s_t$, a cumulative return from $s_t$ denoted as $R_t=\Sigma_{i=t}^{T-1}\gamma^{i-t} r_i$, and $x_{q,t}=[x_t=f_{\phi}(s_t)]_q$, the value of $x_{q,t}$ can be computed via count-based estimation as
\begin{equation}\label{Eq:Cqt}
\begin{aligned}
	C_{q,t}(x_{q,t}) = \frac{1}{{{N_{x_{q,t}}}}}\sum\limits_{j = 1}^{{N_{x_{q,t}}}}R_t^j(x_{q,t})
\end{aligned}
\end{equation}
Here, $N_{x_{q,t}}$ is the visitation count on $x_{q,t}$. However, as an encoder network $f_{\phi}$ is updated during training, the match between a specific state $s$ and $x=f_{\phi}(s)$ can break. Thus, it becomes hard to accurately estimate the value of $s$ via the count-based visit on $x_{q,t}$. To resolve this, we adopt a moving average with a buffer size of $m$ when computing $C_{q,t}(x_{q,t})$ and store the updated value in the extended codebook, $\mathcal{D}_{VQ}$. 
Appendix \ref{app:details_codebook} presents structural details of $\mathcal{D}_{VQ}$.

After constructing $\mathcal{D}_{VQ}$, now we need to determine a goal-reaching trajectory $\tau_{s_t}^*$, defined in Definition \ref{def:goal_reaching_traj}, in the latent space. This trajectory is considered as a reference trajectory to incentivize desired transitions. Let the state sequence from $s_t$ and its corresponding latent sequence projected by $f_{\phi}$ as $\tau_{s_t}=\{s_{t},s_{t+1},s_{t+2},...,s_T\}$ and $\tau_{x_t}=f_{\phi}(\tau_{s_t})$, respectively.
Then, latent sequence after quantization process can be expressed as $\tau_{\chi_t}=[f_{\phi}(\tau_{x_t})]_q=\{x_{q,t},x_{q,t+1},x_{q,t+2},...,x_{q,T}\}$. To evaluate the value the of trajectory $\tau_{\chi_t}$, we use $C_{q,t}$ value in the codebook $\mathcal{D}_{VQ}$ of an initial quantized vector $x_{q,t}$ in $\tau_{\chi_{t}}$.

\begin{figure*}[hbt]
    \centering
    \subfigure[Trajectory embedding.]{\includegraphics[width=0.31\linewidth]{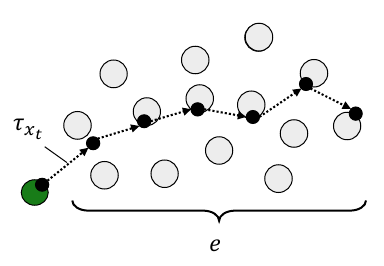}}    
    \subfigure[Trajectory in quantized latent space.]{\includegraphics[width=0.31\linewidth]{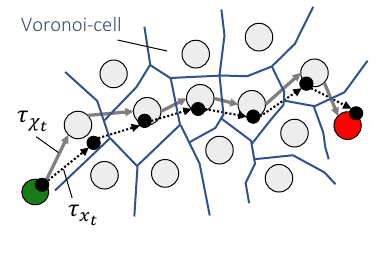}}
    \subfigure[Intrinsic reward generation.]{\includegraphics[width=0.31\linewidth]{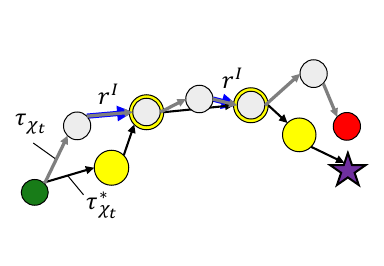}}     
    \vskip -0.10in	
    \caption{Intrinsic reward generation by comparing the current trajectory in quantized latent space ($\tau_{\chi_t}$) with a sampled goal-reaching trajectory ($\tau_{\chi_t}^*$).}
    \label{fig:intrinsic_reward_schematic}
    \vskip -0.15in
\end{figure*}

To encourage desired transitions contained in $\tau_{\chi_t}$ with high $C_{q,t}$ value, we need to keep a sequence data of $\tau_{\chi_t}$. For a given starting node $x_{q,t}$, we keep top-$k$ sequences in $\mathcal{D}_{seq}$ based on their $C_{q,t}$. Thus, $\mathcal{D}_{seq}$ consists of two parts; $\mathcal{D}_{\tau_{\chi_t}}$ stores top-$k$ sequences of $\tau_{\chi_t}$ and $\mathcal{D}_{C_{q,t}}$ stores their corresponding $C_{q,t}$ values. Updating algorithm for a sequence buffer $\mathcal{D}_{seq}$ and structural details of $\mathcal{D}_{seq}$ are presented in Appendix \ref{app:details_codebook}.

As in Definition \ref{def:goal_reaching_traj}, the highest return in cooperative multi-agent tasks can only be achieved when the semantic goal is satisfied. Thus, once agents have achieved a common goal during training, goal-reaching trajectories starting from various initial positions are stored in $\mathcal{D}_{seq}$. 
After we construct $\mathcal{D}_{seq}$, a reference trajectory $\tau_{\chi_t}^*$ can be sampled out of $\mathcal{D}_{\tau_{\chi_t}}$. For a given initial position $x_{q,t}$ in the quantized latent space, we randomly sample a reference trajectory or goal-reaching trajectory from $\mathcal{D}_{seq}$.

\subsection{Intrinsic Reward Generation}
\label{subsec:intrinsic_reward_generation}
With a goal-reaching trajectory $\tau_{\chi_t}^*$ from the current state, we can determine the desired transitions that lead to a goal-reaching path, simply by checking whether the quantized latent $x_{q,t}$ at each timestep $t$ is in $\tau_{\chi_t}^*$. However, before quantized vectors $e$ in the codebook well cover the latent distribution $\chi$, only a few $e$ vectors are selected and thus the same $x_q$ will be repeatedly obtained. In such a case, staying in the same $x_q$ will be encouraged by intrinsic reward if $x_q \in \tau_{\chi_t}^*$. To prevent this, we only provide an incentive to the desired transition toward $x_{q,t+1}$ such that $x_{q,t+1} \in \tau_{\chi_t}^{*} \;\textrm{and}\; x_{q,t+1} \ne {x_{q,t}}$. A remaining problem is how much we incentivize such a desired transition. Instead of an arbitrary incentive, we want to design an additional reward to guarantee a better TD-target, to converge on optimal policy.
\begin{proposition}
\label{proposition_1}
Provided that $\tau_{\chi_t}^{*}$ is a goal-reaching trajectory and $s' \in \tau_{\chi_t}^{*}$, an intrinsic reward $r^I(s'):=\gamma ({C_{q,t}}(s') - {{\max }_{\boldsymbol{a}'}}{Q_{{\theta ^ - }}}(s',\boldsymbol{a}'))$ to the current TD-target $y=r(s,\boldsymbol{a}) + \gamma V_{{\theta ^ - }}(s')$ guarantees a true TD-target as $y^*=r(s,\boldsymbol{a}) + \gamma V^*(s')$, where $V^*(s')$ is a true value of $s'$.   
\end{proposition}
\begin{proof} 
Please see Appendix \ref{sec:proof_proposition1}.
\end{proof}
According to Proposition \ref{proposition_1}, when $\tau_{\chi_t}^{*}$ is a goal-reaching trajectory and $s' \in \tau_{\chi_t}^{*}$, we can set a true TD-target by adding an intrinsic reward $r^I(s')$ to the current TD-target $y$, yielding a better convergence on an optimal policy. In the case when a reference trajectory $\tau_{\chi_t}^{*}$ is not a goal-reaching trajectory, $r^I(s')$ incentivizes the transition toward the high-return trajectory experienced so far. Thus, we define a latent goal-guided intrinsic reward $r^I$ as follows.
\begin{equation}\label{Eq:intrinsic_reward}
\begin{aligned}
r_t^I(s_{t+1})=&\gamma ({C_{q,t}}(s_{t+1}) - {{\max }_{\boldsymbol{a}'}}{Q_{{\theta ^ - }}}(s_{t+1},\boldsymbol{a}')),\\
&\textrm{if} \; x_{q,t+1} \in \tau_{\chi_t}^{*} \;\textrm{and}\; x_{q,t+1} \ne {x_{q,t}}
\end{aligned}
\end{equation}
Note that $r_t^I(s_{t+1})$ is added to $y_t=r_t+\gamma V_{\theta^-}$ not $y_{t+1}$. In addition, we can make sure that $r^I$ becomes non-negative so that an inaccurate estimate of $C_{q,t}(s')$ in the early training phase does not adversely affect the estimation of $V_{\theta^-}$. Algorithm \ref{alg:Goal_traj} summarizes the overall method for goal-reaching trajectory and an intrinsic reward generation. Figure \ref{fig:intrinsic_reward_schematic} illustrates the schematic diagram of quantized trajectory embeddings  $\tau_{\chi_t}$ via VQ-VAE and intrinsic reward generation by comparing it with a goal-reaching trajectory, $\tau_{\chi_t}^{*}$.
\vspace{-0.06in}
\begin{algorithm}[!hbt]    
   \caption{Goal-reaching Trajectory and Intrinsic Reward Generation}
   \label{alg:Goal_traj}
\begin{algorithmic}[0] 
   \STATE {\bfseries Given:} Sequences of the current batch $[\tau_{\chi_t}^i]_{i=1}^{B}$, a sequence buffer $\mathcal{D}_{seq}$, an update interval $n_{\textrm{freq}}$ for $\tau_{\chi_t}^*$, and VQ-VAE codebook $\mathcal{D}_{VQ}$
   \FOR{$i=1$ {\bfseries to} $B$}
   \STATE Compute $R_{t}^i$
   \FOR{$t=0$ {\bfseries to} $T$}
   \STATE Get index $z_t \leftarrow \tau_{\chi_t}^{i}$
   \IF{\textrm{mod}$(t,n_{\textrm{freq}})$}   
   \STATE Run {\bfseries \cref{alg:Update_seq}} to update $\mathcal{D}_{seq}^{z_t}$ with $R_{t}^i$
   \STATE Sample a reference trajectory $\tau_{\chi_t}^*$ from $\mathcal{D}_{seq}^{z_t}$
   \ELSE{}
   \IF{$z_t \in \tau_{\chi_t}^*$ and $z_t \neq z_{t-1}$}
   \STATE Get $C_{q,t}\leftarrow\mathcal{D}_{VQ}^{z_t}.C_{q,t}$    
   \STATE $(r_{t-1}^I)^i\leftarrow\gamma\textrm{max}(C_{q,t}-\textrm{max}_{a'}Q_{\theta^{-}}(s_t,a'),0)$
   \ENDIF
   \ENDIF
   \ENDFOR   
   \ENDFOR   
\end{algorithmic}
\end{algorithm}
\vspace{-0.15in}

\subsection{Overall Learning Objective}
\label{subsec:overall_learning}
This paper adopts a conventional CTDE paradigm \cite{oliehoek2008optimal,gupta2017cooperative}, and thus any form of mixer structure can be used for value factorization. We use the mixer structure presented in QMIX \cite{rashid2018qmix} similar to \cite{yang2022ldsa,wang2020rode,jeon2022maser} to construct the joint Q-value ($Q^{tot}$) from individual Q-functions. By adopting the latent goal-guided intrinsic reward $r^I$ to Eq. \eqref{Eq:CTDE_loss}, the overall loss function for the policy learning can be expressed as follows.
\begin{equation}\label{Eq:overall_loss}
\begin{aligned}
    \mathcal{L}({\theta}) = \left( r^{\textrm{ext}} + r^I + \gamma {{\max }_{{\boldsymbol{a}'}}}{Q_{\theta^{-}}^{tot}}({s'},{\boldsymbol{a}'}) - {Q_{\theta}^{tot}}(s,{\boldsymbol{a}}) \right)^2
\end{aligned}
\end{equation}

\begin{figure*}[!htb]
    \begin{center}
        \centerline{\includegraphics[width=1.0\linewidth]{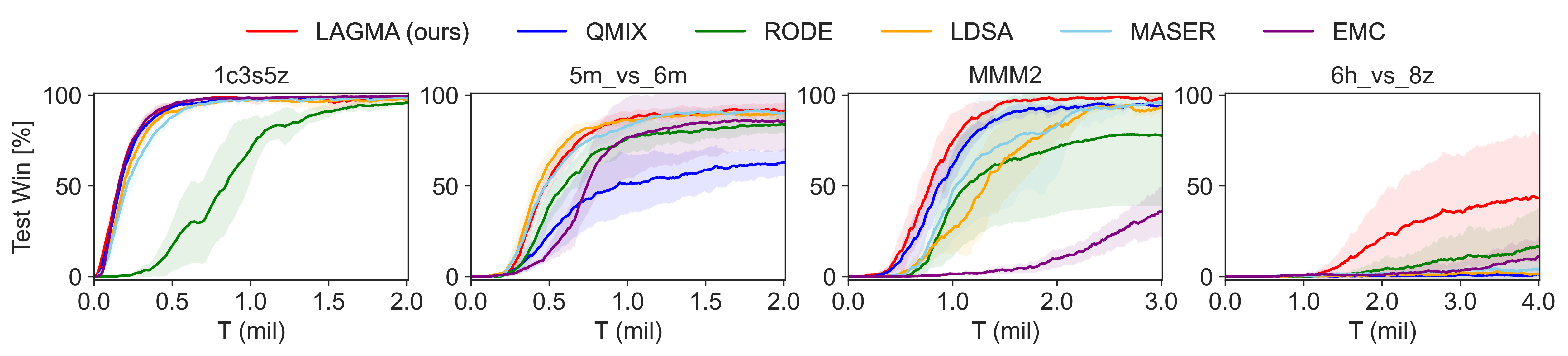}}
        \vskip -0.3in
    \end{center}
    \caption{Performance comparison of LAGMA against baseline algorithms on two {\bf easy and hard} SMAC maps: \texttt{1c3s5z}, \texttt{5m\_vs\_6m}, and two {\bf super hard} SMAC maps: \texttt{MMM2}, \texttt{6h\_vs\_8z}. (Dense reward setting)}
    \label{fig:smac_performance_dense}
    \vskip -0.1in
\end{figure*}	
\begin{figure*}[!htb]
    \begin{center}
        \centerline{\includegraphics[width=1.0\linewidth]{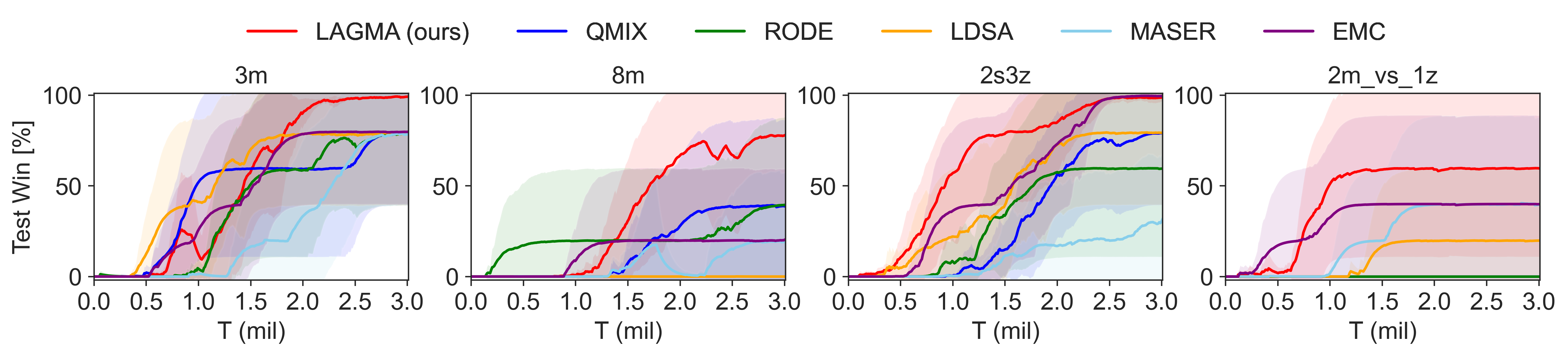}}
        \vskip -0.3in
    \end{center}
    \caption{Performance comparison of LAGMA against baseline algorithms on four maps: \texttt{3m}, \texttt{8m}, \texttt{2s3z}, and \texttt{2m\_vs\_1z}. (Sparse reward setting)}
    \label{fig:smac_performance_sprase}
    \vskip -0.2in
\end{figure*}	

Note that here $r^I$ does not include any scale factor to control its magnitude.
For an individual policy via Q-function, GRUs are adopted to encode a local action-observation history $\tau$ to overcome the partial observability in POMDP similar to most MARL approaches \cite{sunehag2017value,rashid2018qmix,son2019qtran,rashid2020weighted,wang2020qplex}. However, in Eq. \eqref{Eq:overall_loss}, we express the equation with $s$ instead of $\tau$ for the conciseness and coherence with the mathematical derivation. The overall training algorithm for both VQ-VAE training and policy learning is presented in Appendix \ref{app:details_implementation}.

\vspace{-0.05in}
\section{Experiments}
\label{sec:experiments}
In this section, we present experiment settings and results to evaluate the proposed method. We have designed our experiments with the intention of addressing the following inquiries denoted as Q1-3.	
\vspace{-0.1in}
\begin{itemize}
    \item Q1. The performance of LAGMA in comparison to state-of-the-art MARL frameworks in both dense and sparse reward settings
    \vspace{-0.075in}
    \item Q2. The impact of the proposed embedding method on overall performance
    \vspace{-0.075in}
    \item Q3. The efficiency of latent goal-guided incentive compared to other reward design
\end{itemize}
\vspace{-0.1in}
We consider complex multi-agent tasks such as SMAC \cite{samvelyan2019starcraft} and GRF \cite{kurach2020google} as benchmark problems. In addition, as baseline algorithm, we consider various baselines in MARL such as QMIX \cite{rashid2018qmix}, RODE \cite{wang2020rode} and LDSA \cite{yang2022ldsa} adopting a role or skill conditioned policy, MASER \cite{jeon2022maser} presenting agent-wise individual subgoals from replay buffer, and EMC \cite{zheng2021episodic} adopting episodic control. Appendix \ref{app:details_experiments} presents further details of experiment settings and implementations, and Appendix \ref{computation_analysis} illustrates the resource usage and the computational cost required for the implementation and training of LAGMA. In addition, additional generalizability tests of LAGMA are presented in Appendix \ref{generalizability_test}. Our code is available at: \textcolor{blue}{\texttt{https://github.com/aailabkaist/LAGMA}}.

\subsection{Performance evaluation on SMAC}
\textbf{Dense reward settings} \; For dense reward settings, we follow the default setting presented in \cite{samvelyan2019starcraft}.
Figure \ref{fig:smac_performance_dense} illustrates the overall performance of LAGMA. Thanks to quantized embedding and latent goal-guided incentive, LAGMA shows significant performance improvement compared to the backbone algorithm, i.e., QMIX, and other state-of-the-art (SOTA) baseline algorithms, especially in \textbf{super hard} SMAC maps.

\textbf{Sparse reward settings} \; For a sparse reward setting, we follow the reward design in MASER \cite{jeon2022maser}. Appendix \ref{app:details_experiments} enumerates the details of reward settings.
Similar to dense reward settings, LAGMA shows the best performance in sparse reward settings thanks to the latent goal-guided incentive. Sparse reward hardly generates a reward signal in experience replay, thus training with the experience of the exact same state takes a long time to find the optimal policy. However, LAGMA considers the value of semantically similar states projected onto the same quantized vector during training, so its learning efficiency is significantly increased.

\begin{figure}[!hbt]
    \begin{center}
        \centerline{\includegraphics[width=1.0\linewidth]{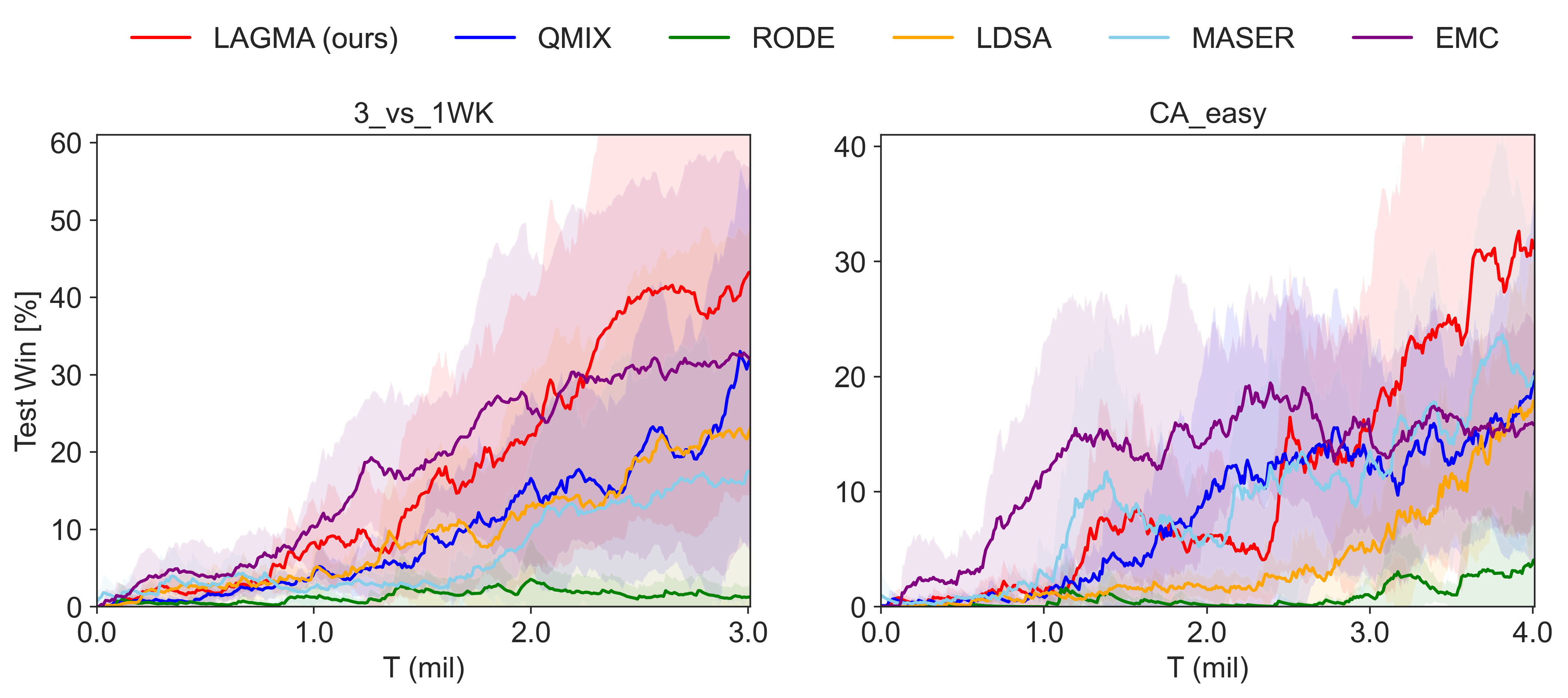}}
        \vspace{-0.3in}
    \end{center}
    \caption{Performance comparison of LAGMA against baseline algorithms on two GRF maps: \texttt{3\_vs\_1WK} and \texttt{CounterAttack(CA)\_easy}. (Sparse reward setting)}
    \label{fig:smac_performance_sprase}
    \vspace{-0.15in}
\end{figure}

\begin{figure*}[!hbt]
    \begin{center}
        \centerline{\includegraphics[width=0.8\linewidth]{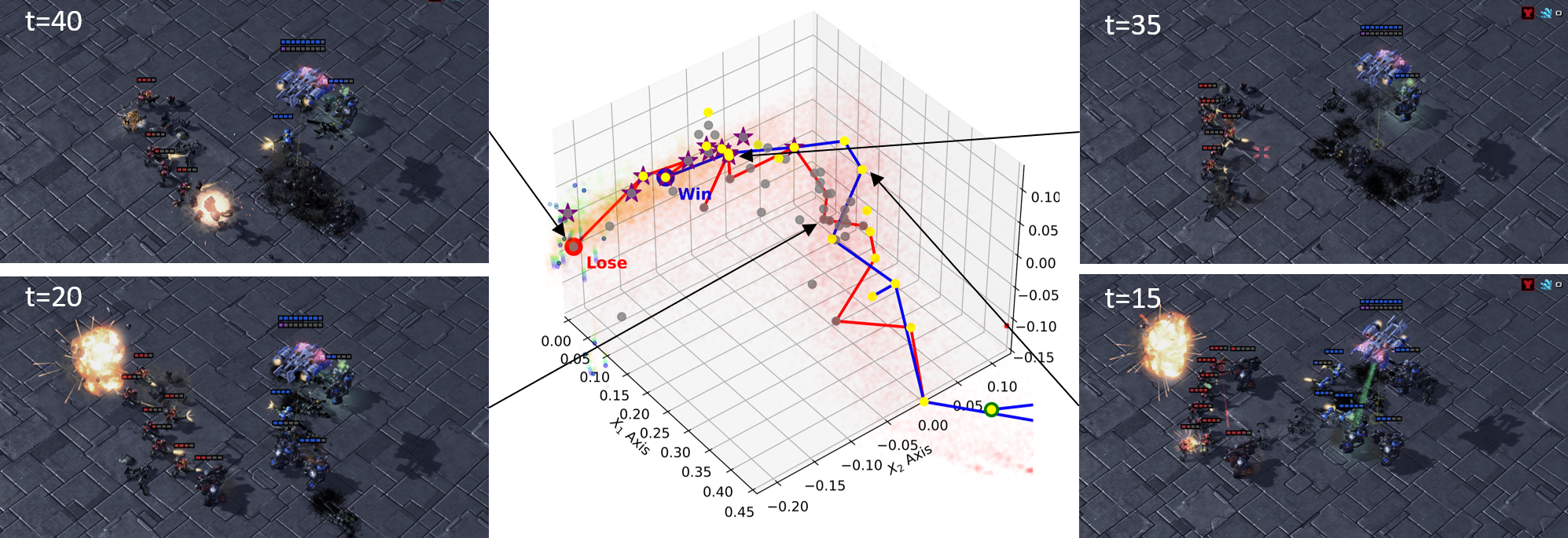}}
        \vspace{-0.2in}
    \end{center}
    \caption{Qualitative analysis on SMAC \texttt{MMM2} (red teams are RL-agents). Purple stars represent quantized embeddings of goal states in replay buffer $\mathcal{D}$. Yellow dots indicate the quantized embeddings in a sampled goal-reaching trajectory starting from an initial state denoted by a green dot. Gray dots and transparent dots are the same as Figure \ref{fig:vqvae_embedding}. Blue and red dots indicate terminal embeddings of two trajectories, respectively.}
    \label{fig:trj_qualitative}
    \vspace{-0.15in}
\end{figure*}	

\subsection{Performance evaluation on GRF}
Here, we conduct experiments on additional sparse reward tasks in GRF to compare LAGMA with baseline algorithms. For experiments, we do not utilize any additional algorithm for sample efficiency such as prioritized experience replay \cite{schaul2015universal} for all algorithms.

EMC \cite{zheng2021episodic} shows comparable performance by utilizing an episodic buffer, which benefits in generating a positive reward signal via additional episodic control term. However, LAGMA with a modified VQ codebook could guide a scoring policy without utilizing an additional episodic buffer as being required in EMC. Therefore, LAGMA achieves a similar or better performance with less memory requirement.

\subsection{Ablation study}
In this subsection, we conduct ablation studies to see the effect of the proposed embedding method and latent goal-guided incentive on overall performance. 
\begin{figure}[!b]
    \vskip -0.15in
    \begin{center}
        \centerline{\includegraphics[width=1.0\linewidth]{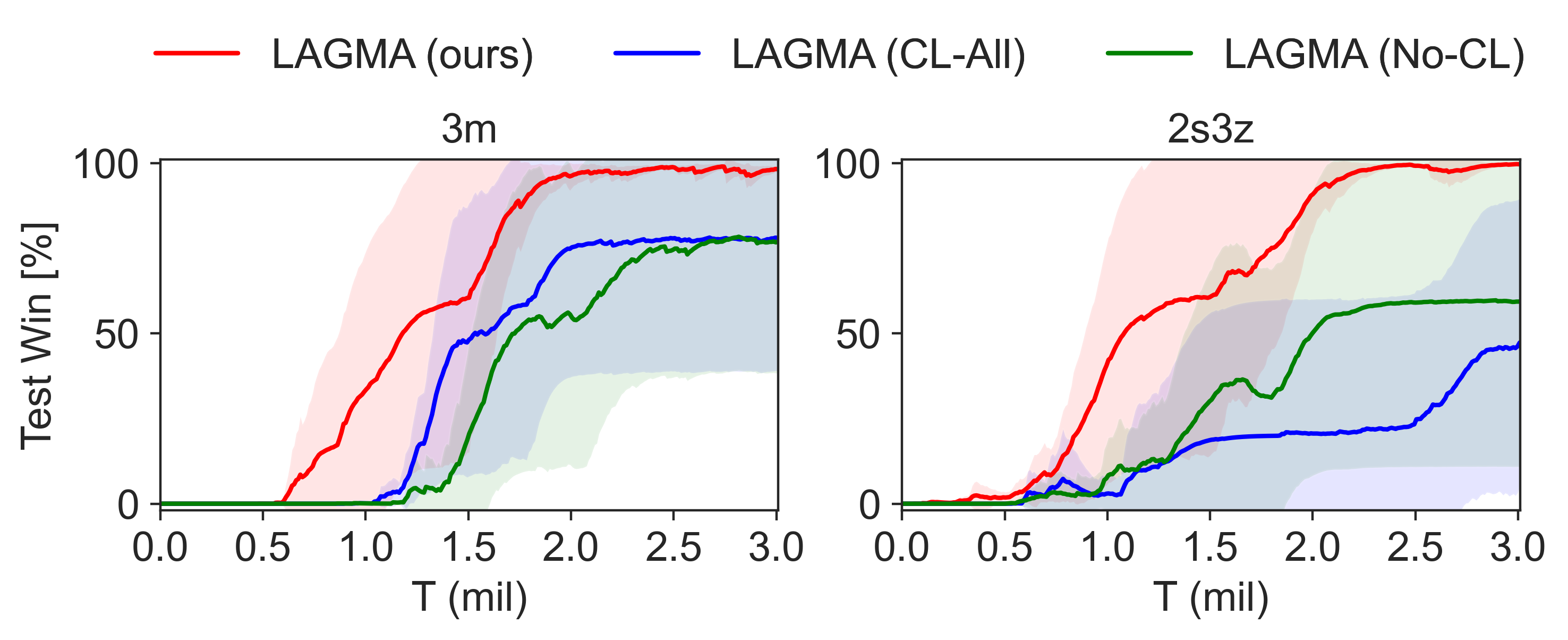}}
       \vspace{-0.3in}
    \end{center}
    \caption{Ablation study considering the coverage loss (CL) on four SMAC maps: \texttt{3m} and \texttt{2s3z}. (Sparse reward setting)}
    \label{fig:smac_ablation}
    \vspace{-0.1in}
\end{figure}	
\begin{figure}[!b]
    \vskip -0.1in
    \begin{center}
        \centerline{\includegraphics[width=1.0\linewidth]{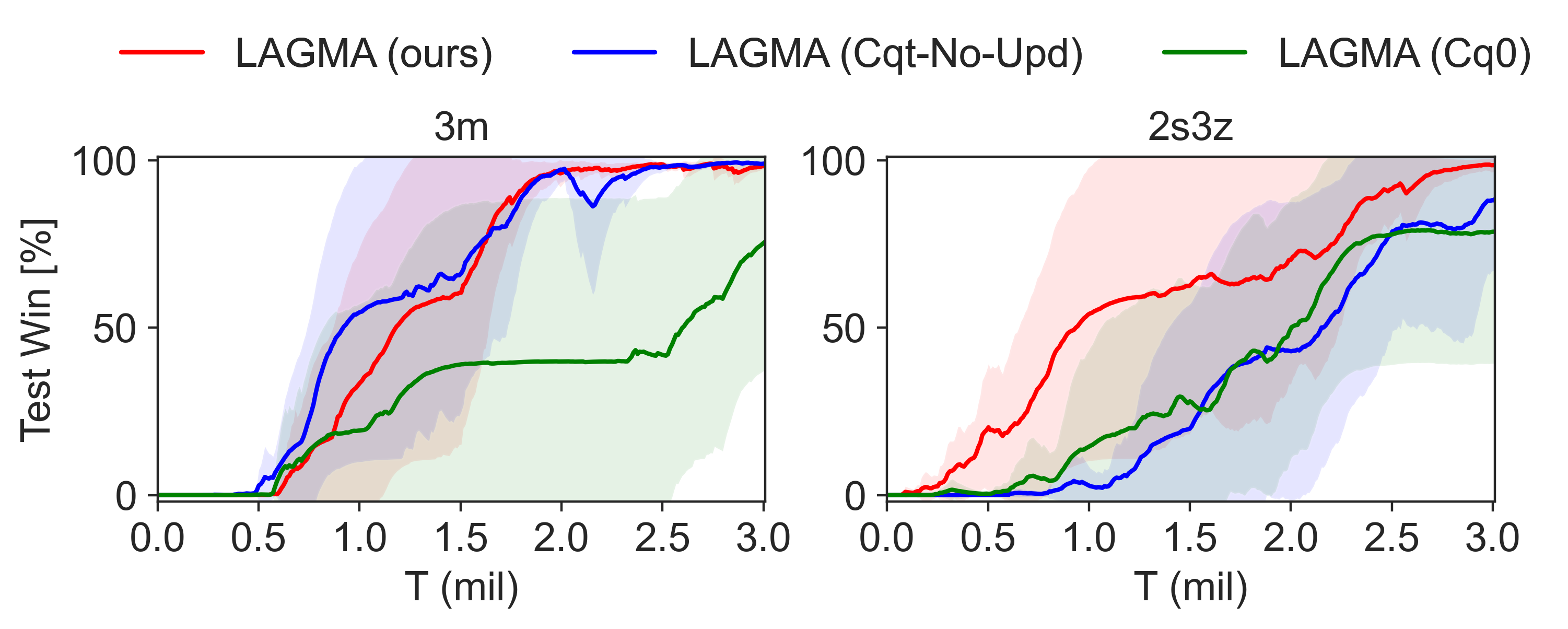}}
        \vspace{-0.3in}
    \end{center}
    \caption{Performance comparison of goal-guided incentive with other reward design choices on two SMAC maps: \texttt{3m} and \texttt{2s3z}. (Sparse reward setting)}
    \label{fig:smac_ablation_Cqt}
\end{figure}
We compare LAGMA (ours) with ablated configurations such as LAGMA (CL-All) trained with $\lambda_{\textrm{cvr}}^{\textrm{all}}$ considering all quantized vectors at each timestep and LAGMA (No-CL) trained without coverage loss.

Figure \ref{fig:smac_ablation} illustrates the effect of the proposed coverage loss in modified VQ-VAE on the overall performance. As shown in Fig. \ref{fig:smac_ablation}, the performance decreases when the model is trained without coverage loss or trained with $\lambda_{\textrm{cvr}}^{\textrm{all}}$ instead of $\lambda_{\textrm{cvr}}$. The results imply that, without the proposed coverage loss, quantized latent vectors may not cover $\chi$ properly and thus $x_q$ can hardly represent the projected states. As a result, a goal-reaching trajectory that consists of a few quantized vectors yields no incentive signal in most transitions.
	
In addition, we conduct an ablation study on reward design. We consider a sum of undiscounted rewards, $C_{q0}=\Sigma_{t=0}^{T-1} r_t$, for trajectory value estimation instead of $C_{q,t}$, denoted as \textbf{LAMGA (Cq0)}. We also consider the LAMGA configuration with goal-reaching trajectory generation only at the initial state denoted by \textbf{LAMGA (Cqt-No-Upd)}. Figure \ref{fig:smac_ablation_Cqt} illustrates the results.
Figure \ref{fig:smac_ablation_Cqt} implies that the reward design of $C_{q,t}$ shows a more stable performance than both \textbf{LAMGA (Cq0)} and \textbf{LAMGA (Cqt-No-Upd)}.

\subsection{Qualitative analysis} 
In this section, we conduct a qualitative analysis to observe how the states in an episode are projected onto quantized vector space and receive latent goal-guided incentive compared to goal-reaching trajectory sampled from $\mathcal{D}_{seq}$.
Figure \ref{fig:trj_qualitative} illustrates the quantized embedding sequences of two trajectories: one denoted by a blue line representing a battle-won trajectory and the other denoted by a red line representing a losing trajectory. In Fig. \ref{fig:trj_qualitative}, a losing trajectory initially followed the optimal sequence denoted by yellow dots but began to bifurcate at $t=20$ by losing Medivac and two more allies. Although the losing trajectory still passed through goal-reaching sequences during an episode, it ultimately reached a terminal state without a chance to defeat the enemies at $t=40$, as indicated by the absence of a purple star.
On the other hand, a trajectory that achieved victory followed the goal-reaching path and reached a goal state at the end, as indicated by purple stars. Since only transitions toward the sequences on the goal-reaching path are incentivized, LAGMA can efficiently learn a goal-reaching policy, i.e., the optimal policy in cooperative MARL.

\section{Conclusions}
\label{sec:conclusions}
This paper presents LAGMA, a framework to generate a goal-reaching trajectory in latent space and a latent goal-guided incentive to achieve a common goal in cooperative MARL. Thanks to the quantized embedding space, the experience of semantically similar states is shared by states projected onto the same quantized vector, yielding efficient training. The proposed latent goal-guided intrinsic reward encourages transitions toward a goal-reaching trajectory. 
Experiments and ablation studies validate the effectiveness of LAGMA. 

\section*{Acknowledgements}
This research was supported by AI Technology Development for Commonsense Extraction, Reasoning, and Inference from Heterogeneous Data (IITP) funded by the Ministry of Science and ICT (2022-0-00077).

\section*{Impact Statement}
This paper primarily focuses on advancing the field of Machine Learning through multi-agent reinforcement learning. While there could be various potential societal consequences of our work, none of which we believe must be specifically highlighted here.
\nocite{langley00}

\bibliography{lagma}
\bibliographystyle{icml2024}

\newpage
\appendix
\onecolumn

\section{Mathematical Proof}
\label{sec:proof_proposition1}
Here, we provide the omitted proof of Proposition \ref{proposition_1}.
\begin{proof} 
Let $y=r(s,\boldsymbol{a}) + \gamma V_{{\theta ^ - }}$ be the current TD-target with the target network parameterized with $\theta^-$. Adding an intrinsic reward $r^I(s')$ to $y$ yields $y'=y+r^I(s')$. Now, we need to check whether $y'$ accurately estimates $y^*=r+\gamma V^*(s')$.
\begin{equation}\label{Eq:proof_proposition1}
    \begin{aligned}
        \mathbb{E}[y']&=\mathbb{E}[r(s,\boldsymbol{a}) + \gamma V_{{\theta ^ - }}+r^I(s')] \\ 
        &=\mathbb{E}[r(s,\boldsymbol{a}) + \gamma \textrm{max}_{\boldsymbol{a'}}Q_{{\theta^-}}(s',\boldsymbol{a'})+\gamma(C_{q,t}(s')-\textrm{max}_{\boldsymbol{a'}}Q_{{\theta^-}}(s',\boldsymbol{a'}))] \\ 
        &=\mathbb{E}[r(s,\boldsymbol{a}) + \gamma(C_{q,t}(s'))]\\
        &=\mathbb{E}[r(s,\boldsymbol{a}) + \gamma (\mathbb{E}[\sum_{i=t+1}^{T-1}\gamma^{i-(t+1)} r_i])]\\
        &=r(s,\boldsymbol{a}) + \gamma \mathbb{E}[r_{t+1}+\gamma r_{t+2} + \cdots + \gamma^{T-t-2}r_{T-1}]\\
        &=r(s,\boldsymbol{a}) + \gamma V^*(s')\\
    \end{aligned}    
\end{equation}
The last equality in Eq. \eqref{Eq:proof_proposition1} holds since $s'$ is on a goal-reaching trajectory, i.e., $s' \in \tau_{\chi_0}^{*}$ whose return is maximized, and $\mathbb{E}[r_{t+1}+\gamma r_{t+2} + \cdots$] is an unbiased Monte-Carlo estimate of $V^*(s')$.
\end{proof}

\section{Experiment Details}
\label{app:details_experiments}
In this section, we present details of SMAC \cite{samvelyan2019starcraft} and GRF \cite{kurach2020google}, and we also list hyperparemeter settings of LAGMA for each task.
Tables \ref{tab:details of SMAC} and \ref{tab:details of GRF} present the dimensions of state and action spaces and the maximum episodic length.

\begin{table}[htbp]
		\centering
		\caption{Dimension of the state space and the action space and the episodic length of SMAC}
		\vspace{0.3cm}
		\begin{tabular}{cccc}
			\toprule
			Task & Dimension of state space & Dimension of action space & Episodic length\\
            \midrule 
            3m & 48 & 9 & 60\\
            8m & 168 & 14 & 120\\
            2s3z & 120 & 11 & 120\\
            2m\_vs\_1z & 26 & 7 & 150\\
            1c3s5z & 270 & 15 & 180\\
            5m\_vs\_6m & 98 & 12 & 70\\
            MMM2 & 322 & 18 & 180\\
            6h\_vs\_8z & 140 & 14 & 150\\
			\bottomrule
		\end{tabular}\label{tab:details of SMAC}
\end{table}

\begin{table}[htbp]
    \centering
    \caption{Dimension of the state space and the action space and the episodic length of GRF}
    \vspace{0.3cm}
    \begin{tabular}{cccc}
        \toprule
            Task & Dimension of state space & Dimension of action space & Episodic length\\
            \midrule
            3\_vs\_1WK & 26 & 19 & 150\\
            CA\_easy & 30 & 19 & 150\\
        \bottomrule
    \end{tabular}\label{tab:details of GRF}
\end{table} 

In addition, Table \ref{tab:hyper_settings} presents the task-dependent hyperparameter settings for all experiments. As seen from Table \ref{tab:hyper_settings}, we used similar hyperparameters across various tasks. For an update interval $n_{\textrm{freq}}$ in Algorithm \ref{alg:Goal_traj}, we use the same value $n_{\textrm{freq}}=5$ for all experiments. $\epsilon_T$ represents annealing time for exploration rate of $\epsilon$-greedy, from 1.0 to 0.05. 

After some parametric studies, adjusting hyperparameter for VQ-VAE training such as  $n_{\textrm{freq}}^{cd}$ and $n_{\textrm{freq}}^{vq}$, instead of varying $\lambda$ values listed as $\lambda_{\textrm{vq}}$, $\lambda_{\textrm{commit}}$, and $\lambda_{\textrm{cvr}}$, provides more efficient way of searching parametric space. 
Thus, we primarily adjust $n_{\textrm{freq}}^{cd}$ and $n_{\textrm{freq}}^{vq}$ according to tasks, while keeping the ratio between $\lambda$ values the same. 

For hyperparameter settings, we recommend the efficient bounds for each hyperparameter based on our experiments as follows:
\begin{itemize}
    \item Number of codebook, $n_c$: 64-512
    \item Update interval for VQ-VAE, $n_{\textrm{freq}}^{vq}$: 10-40 (under $n_{\textrm{freq}}^{cd}=10$)
    \item Update interval for extended codebook, $n_{\textrm{freq}}^{cd}$: 10-40 (under $n_{\textrm{freq}}^{vq}=10$)
    \item Number of reference trajectory, $k$: 10-30             
    \item Scale factor of coverage loss, $\lambda_{\textrm{cvr}}$: 0.25-1.0 (under $\lambda_{\textrm{vq}}=1.0$ and $\lambda_{\textrm{commit}}=0.5$)            
\end{itemize}

Note that larger values of $n_c$ and $k$, and smaller values of $n_{\textrm{freq}}^{vq}$ and $n_{\textrm{freq}}^{cd}$ will increase the computational cost. 

\begin{table}[!htbp]
    \centering
    \caption{Hyperparameter settings for experiments.}
    \vspace{0.1cm}
    \label{tab:hyper_settings}
\adjustbox{max width=\textwidth}{
\begin{tabular}{cccccccccc}
\toprule
\multicolumn{2}{c}{task}                         &  $n_c$ & $D$ & $\lambda_{\textrm{vq}}$ & $\lambda_{\textrm{commit}}$ 
            & $\lambda_{\textrm{cvr}}$ & $\epsilon_T$ & $n_{\textrm{freq}}^{cd}$ & $n_{\textrm{freq}}^{vq}$\\
            \midrule
\multirow{4}{*}{SMAC (sparse)} & 3m         & 256     & 8          & 2.0        & 1.0            & 1.0         & 50K  & 10  & 40    \\
                               & 8m         & 256     & 8          & 2.0        & 1.0            & 1.0         & 50K  & 20  & 10    \\
                               & 2s3z       & 256     & 8          & 2.0        & 1.0            & 1.0         & 50K  & 10  & 40    \\
                               & 2m\_vs\_1z & 256     & 8          & 2.0        & 1.0            & 1.0         & 500K & 20  & 10    \\
                               \midrule
\multirow{4}{*}{SMAC (dense)}  & 1c3s5z     & 64      & 8          & 1.0        & 0.5            & 0.5         & 50K  & 40  & 10    \\
                               & 5m\_vs\_6m & 64      & 8          & 1.0        & 0.5            & 0.5         & 50K  & 40  & 10    \\
                               & MMM2       & 64      & 8          & 1.0        & 0.5            & 0.5         & 50K  & 40  & 10    \\
                               & 6h\_vs\_8z & 256     & 8          & 2.0        & 1.0            & 1.0         & 500K & 40  & 10    \\
                               \midrule
\multirow{2}{*}{GRF (sparse)}  & 3\_vs\_1WK & 256     & 8          & 2.0        & 1.0            & 1.0         & 50K & 20  & 10    \\
                               & CA\_easy   & 256     & 8          & 2.0        & 1.0            & 1.0         & 50K & 10  & 20    \\
                               \bottomrule
\end{tabular}
}
\end{table}

Table \ref{tab:smac_sparse_reward} presents the reward settings for SMAC (sparse) which follows the sparse reward settings from \cite{jeon2022maser}.

\begin{table}[!htbp]
    \centering
    \caption{Reward settings for SMAC (sparse)}
    \vspace{0.1cm}
    \label{tab:smac_sparse_reward}
\begin{tabular}{lcl}
\toprule
\multicolumn{1}{c}{Condition} & Sparse reward &  \\
\midrule
All enemies die (Win)         & +200           &  \\
Each enemy dies               & +10            &  \\
Each ally dies                & -5            &  \\
\bottomrule
\end{tabular}
\end{table}


\newpage
\section{Additional Related Works}
\label{app:additional_related}

\textbf{Goal-conditioned RL vs. Subtask conditioned MARL} \;
In a single agent case, Goal-conditioned RL (GCRL) which aims to solve multiple tasks to reach given target-goals has been widely adopted in various tasks including tasks with a sparse reward \cite{kaelbling1993learning,schaul2015universal,andrychowicz2017hindsight}.
GCRL often utilizes the given goal as an additional input to action policy in addition to states \cite{schaul2015universal}. Especially, goal-conditioned hierarchical reinforcement learning (HRL) \cite{kulkarni2016hierarchical,zhang2020generating,chane2021goal} adopts hierarchical policy structure where an upper-tier policy determines subgoal or landmark and a lower-tier policy takes action based on both state and selected a subgoal or landmark.

As one technique, reaching to subgoals generates a reward signal via hindsight experience replay \cite{andrychowicz2017hindsight}, and thus the goal-conditioned policy learn policy to reach the final goal with the help of these intermediate signals. Thus, many researchers \cite{nasiriany2019planning,zhang2020generating,chane2021goal,kim2023imitating,lee2023cqm} have studied on how to generate intermediate subgoals to reach final goals.

In the field of MARL, a subtask\cite{yang2022ldsa}, role\cite{wang2020roma,wang2020rode} or skill\cite{yang2019hierarchical,liu2022heterogeneous} conditioned policy adopted in a hierarchical MARL structure has a structural commonality with a goal-conditioned RL in that lower-tier policy network use designated subtask by the upper-tier network as an additional input when determining individual action. In MARL tasks, such subtasks, roles, or skills are a bit different from subgoals in GCRL, as they are adopted to decompose action space for efficient training or for subtask-dependent coordination. Another major difference is that in a general MARL task, the final goal is not defined explicitly unlike a goal-conditioned RL. MASER \cite{jeon2022maser} adopts the subgoal generation scheme from goal-conditioned RL when it generates an intrinsic reward based on the Euclidean distance between actionable representations \cite{ghosh2018learning} of the current and subgoal observation. 
However, this signal does not guarantee the consistency with learning signal for the joint Q-function. In contrast to MASER, we adopt a latent goal-guided incentive during a centralized training phase based on whether visiting on the promising subgoals or goals in the latent space. Also, the generated incentive by LAGMA theoretically guarantees a better TD-target, yielding better convergence on the optimal policy.

\section{Structure of Extended VQ Codebook}
\label{app:details_codebook}

To compute $C_{q,t}$ via a moving average, data is stored in a FIFO (First in, First Out) style to the codebook $\mathcal{D}_{VQ}$, similar to a replay buffer $\mathcal{D}$. After computing $C_{q,t}(x_{q,t})$ with the current $R_t$, we update the value of $C_{q,t}(x_{q,t})$ in $\mathcal{D}_{VQ}$ as $\mathcal{D}_{VQ}^{z_t}.C_{q,t}\leftarrow C_{q,t}(x_{q,t})$ where $z_t$ is an index of a quantized vector $x_{q,t}$. 

\begin{figure*}[htb]
\vskip 0.2in
\begin{center}
\centerline{\includegraphics[width=\linewidth]{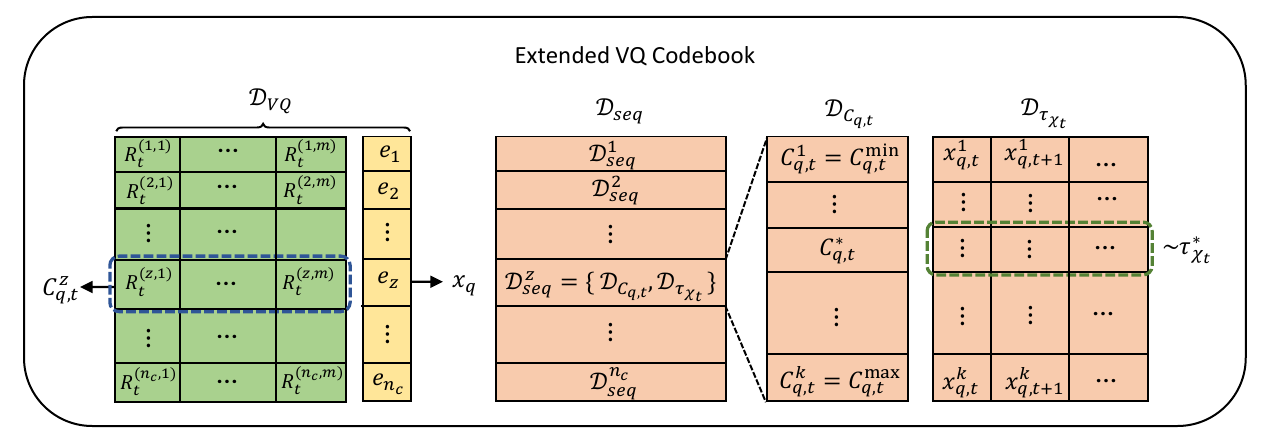}}
\caption{VQ codebook structure.}
\label{fig:VQ_codebook_structure}
\end{center}
\vskip -0.2in
\end{figure*}

\begin{algorithm}[!hbt]
   \caption{Update Sequence Buffer $\mathcal{D}_{seq}$}
   \label{alg:Update_seq}
\begin{algorithmic}[1] 
   \STATE \textcolor{gray}{$\mathcal{D}_{seq}$ keep top $k$ trajectory sequences based on their $C_{q,t}$}
   \STATE {\bfseries Input:} A total reward sum $R_{t}$ and a sequence $\tau_{\chi_t}$
   \STATE Get an initial index $z_t \leftarrow \tau_{\chi_t}[0]$ 
   \STATE Get $\mathcal{D}_{C_{q,t}}, \mathcal{D}_{\tau_{\chi_t}}$ from $\mathcal{D}_{seq}^{z_t}$   
   \IF{$|\mathcal{D}_{C_{q,t}}| < k$} 
   \STATE \texttt{heap\_push}($\mathcal{D}_{C_{q,t}},\mathcal{D}_{\tau_{\chi_t}},C_{q,t},\tau_{\chi_t}$)
   \ELSE{}   
   \STATE $C_{q,t}^\textrm{min} \leftarrow \mathcal{D}_{C_{q,t}}[0]$
   \IF{$R_{t} > C_{q,t}^\textrm{min}$}
   \STATE \texttt{heap\_replace}($\mathcal{D}_{C_{q,t}},\mathcal{D}_{\tau_{\chi_t}},R_{t},\tau_{\chi_t}$)
   \ENDIF
   \ENDIF      
\end{algorithmic}
\end{algorithm}

In Algorithm \ref{alg:Update_seq}, \texttt{heap\_push} and \texttt{heap\_replace} adopt the conventional heap space management rule, with a computational complexity of $\mathcal{O}(\textrm{log}{k})$. The difference is that we additionally store the sequence information in $\mathcal{D}_{\tau_{\chi_t}}$ according to their $C_{q,t}$ values.

\newpage
\section{Implementation Details}
\label{app:details_implementation}
In this section, we present further details of the implementation for LAGMA. Algorithm \ref{alg:time_indexing} presents the pseudo-code for a timestep dependent indexing $\mathcal{J}(t)$ used in Eq. 
\eqref{Eq:coverage_loss}. The purpose of a timestep dependent indexing $\mathcal{J}(t)$ is to distribute the quantized vectors throughout an episode. Thus, Algorithm \ref{alg:time_indexing} tries to uniformly distribute quantized vectors according to the maximum batch time of an episode. By considering the maximum batch time of each episode, Algorithm \ref{alg:time_indexing} can adaptively distribute quantized vectors.

\begin{algorithm}[!hbt]
   \caption{Compute $\mathcal{J}(t)$}
   \label{alg:time_indexing}
\begin{algorithmic}[1]    
   \STATE {\bfseries Input:} For given the maximum batch time $T$, the number of codebook $n_c$, and the current timestep $t$
   
   \IF{$t==0$} 
       \STATE $d=\lfloor n_c/T \rfloor $ 
       \STATE $r=n_c\; \texttt{mod}\;T$ 
       \STATE $i_s = d_n \times T$
       \STATE Keep the values of $d,r,i_s$ until the end of the episode
   \ENDIF       
   \IF{$d \geq 1$} 
       \STATE $\mathcal{J}(t)=d\times t:1:d\times(t+1)$
       \IF{$t < r$} 
            \STATE Append $\mathcal{J}(t)$ with $i_s+t$ 
       \ENDIF
   \ELSE
        \STATE $\mathcal{J}(t)=\lfloor t\times n_c/T \rfloor $
   \ENDIF
   
   \STATE return $\mathcal{J}(t)$
\end{algorithmic}
\end{algorithm}

For given the maximum batch time $t_\textrm{max}$ and the number of codebook $n_c$, \texttt{Line\#4} and \texttt{Line\#5} in Algorithm \ref{alg:time_indexing} compute the quotient and remainder, respectively.
\texttt{Line\#8} compute an array with increasing order starting from the index $d \times t$ to $d \times (t+1)$. \texttt{Line\#10} additionally designate the remaining quantized vectors to the early time of an episode.

Algorithm \ref{alg:train_VQVAE} presents training algorithm to update encoder $f_{\phi}$, decoder $f_{\psi}$, and quantized embeddings $\boldsymbol{e}$ in VQ-VAE. In Algorithm \ref{alg:train_VQVAE}, we also present a separate update for $\mathcal{D}_{VQ}$, which estimates the value of each quantized vector in VQ-VAE.
In addition, the overall training algorithm including training for VQ-VAE is presented in Algorithm \ref{alg:training_lagma}. 

\begin{algorithm}[!hbt]
   \caption{Training algorithm for VQ-VAE and $\mathcal{D}_{VQ}$}
   \label{alg:train_VQVAE}
\begin{algorithmic}[0] 
   \STATE {\bfseries Parameter:} learning rate $\alpha$ and batch-size $B$
   \STATE {\bfseries Input:} $B$ sample trajectories $[\mathcal{T}]_{i=1}^{B}$ from replay buffer $\mathcal{D}$, the current episode number $n_{\textrm{epi}}$, an update interval $n_{\textrm{freq}}^{vq}$ for VQ-VAE and $n_{\textrm{freq}}^{cd}$ for $\mathcal{D}_{VQ}$ update interval.
   
   \FOR{$t=0$ {\bfseries to} $T$}
   \FOR{$i=1$ {\bfseries to} $B$}   
   \IF{\textrm{mod}$(n_{\textrm{epi}},n_{\textrm{freq}}^{cd})$}   
   \STATE Get $R_t^i$ and update $\mathcal{D}_{VQ}$ with Eq. \eqref{Eq:Cqt}.
   \ENDIF
   \IF{\textrm{mod}$(n_{\textrm{epi}},n_{\textrm{freq}}^{vq})$}   
   \STATE Get current state $s_t\sim[\mathcal{T}]_{i=1}$ and compute $\mathcal{J}(t)$ via Algorithm \ref{alg:time_indexing}
   \STATE Compute $\mathcal{L}_{VQ}^{tot}$ via Eq. \eqref{Eq:overall_obj_vqvae} with $f_{\phi}$, $f_{\psi}$, and $\boldsymbol{e}$.
   \ENDIF
   \ENDFOR   
   \ENDFOR   
   \IF{\textrm{mod}$(n_{\textrm{epi}},n_{\textrm{freq}}^{vq})$}   
   \STATE Update $\phi\leftarrow \phi-\alpha \frac{\partial\mathcal{L}_{VQ}^{tot}}{\partial\phi}$,\;$\psi\leftarrow \psi-\alpha \frac{\partial\mathcal{L}_{VQ}^{tot}}{\partial\psi}$,\;$\boldsymbol{e} \leftarrow \boldsymbol{e}-\alpha \frac{\partial\mathcal{L}_{VQ}^{tot}}{\partial\boldsymbol{e}}$
   \ENDIF
\end{algorithmic}
\end{algorithm}

\begin{algorithm}[htbp]
		\begin{center}
			\begin{algorithmic}[1]		
                \STATE {\bfseries Parameter:}  Batch size $B$ and the maximum training time $T_{env}$
                \STATE {\bfseries Input:} $Q_{\theta}^i$ is individual Q-network of $n$ agents, replay buffer $\mathcal{D}$, extended VQ codebook $\mathcal{D}_{VQ}$, and sequence buffer $\mathcal{D}_{seq}$	
				\STATE Initialize network parameters $\theta$, $\phi$, $\psi$, $\boldsymbol{e}$
				\WHILE{$t_{env} \leq T_{env}$}
				\STATE Interact with the environment via $\epsilon$-greedy policy based on $[Q_{\theta}^i]_{i=1}^{n}$ and get a trajectory $\mathcal{T}$		
				\STATE Append $\mathcal{T}$ to  $\mathcal{D}$
				\STATE Get $B$ sample trajectories $[\mathcal{T}]_{i=1}^{B} \sim \mathcal{D}$ 
				\STATE For a given $[\mathcal{T}]_{i=1}^{B}$, run MARL training algorithm and Algorithm \ref{alg:Goal_traj} to update $\theta$ with Eq. \eqref{Eq:overall_loss}, and Algorithm \ref{alg:train_VQVAE} to update $\phi$, $\psi$, $\boldsymbol{e}$ with Eq. \eqref{Eq:overall_obj_vqvae}
								
				\ENDWHILE
			\end{algorithmic}
		\end{center}
		\caption{Training algorithm for LAGMA.}
		\label{alg:training_lagma}
\end{algorithm}

  \newpage
  \section{Generalizability of LAGMA}
  \label{generalizability_test}
  \subsection{Policy robustness test} 
    To assess the robustness of policy learned by our model, we designed tasks with the same unit configuration but highly varied initial positions, ones that agents had not encountered during training, i.e., unseen maps. With these settings, opponent agents will also experience totally different relative positions and thus will make different decisions. We set the different initial positions for this evaluation as follows.

     \begin{figure*}[!hbt]
        \centering
        \includegraphics[width=0.5\linewidth]{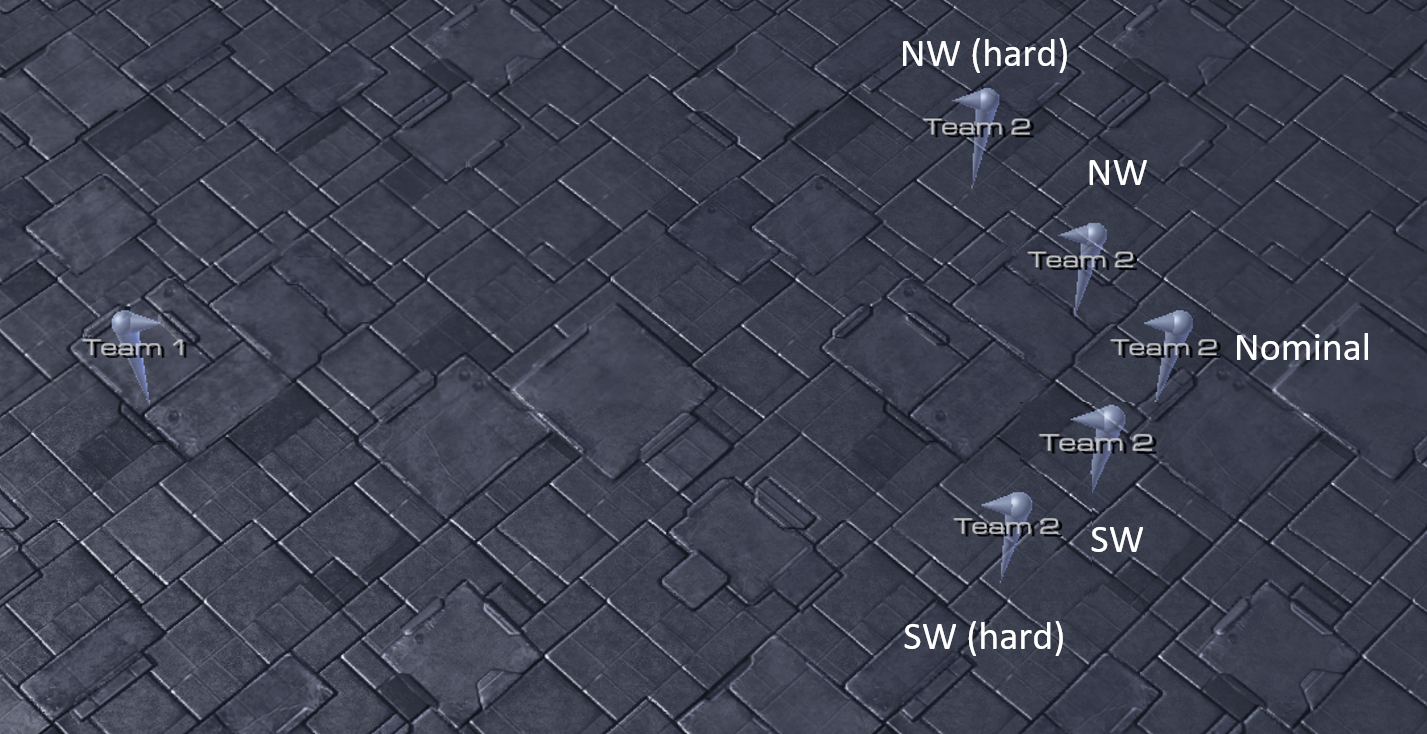}
        \caption{Problem settings for policy robustness test. Team 1 represents the initial position of RL agents, while Team 2 is the initial position of opponents.}
        \label{fig:robustenss_problem}
    \end{figure*}

    As illustrated in Figure \ref{fig:robustenss_problem}, the initial position of each task is significantly moved from the nominal position experienced during the training phase.
    
    For the comparison, we conduct the same experiment for other baselines, such as QMIX \cite{rashid2018qmix} and LDSA \cite{yang2022ldsa}. The model by each algorithm is trained for $T_{env}=1M$ in nominal MMM2 map (denoted as Nominal) and then evaluated under various problem settings, such as NW(hard), NW, SW, and SW(hard). Each evaluation is conducted for 5 different seeds with 32 tests and Table \ref{tab:robustness_test} shows the mean and variance of winrate of each case.
    
    \begin{table}[!htbp]
    \centering
    \caption{Policy robustness test on SMAC MMM2 (super hard). All models are trained for $T_{env}=1M$. The percentage (\%) in parentheses represents the ratio compared to a nominal mean value.}
    \adjustbox{max width=\textwidth}{
    \begin{tabular}{lccccc}
    \toprule
    \multicolumn{1}{c}{} & NW(hard) & NW    & Nominal & SW    & SW(hard) \\ 
    \midrule
    LAGMA                & 0.275 $\pm$ 0.064 (28.2\%)   & 0.500 $\pm$ 0.104 (51.3\%) & 0.975 $\pm$ 0.026   & 0.556 $\pm$ 0.051 (57.1\%) & 0.394 $\pm$ 0.042 (40.4\%) \\
    QMIX                 & 0.050 $\pm$ 0.036 (13.1\%)   & 0.138 $\pm$ 0.100 (36.1\%) & 0.381 $\pm$ 0.078   & 0.194 $\pm$ 0.092 (50.8\%) & 0.156 $\pm$ 0.058 (41.0\%) \\
    LDSA                 & 0.000 $\pm$ 0.000 ( 0.0\%)   & 0.081 $\pm$ 0.047 (18.3\%) & 0.444 $\pm$ 0.107   & 0.063 $\pm$ 0.049 (14.1\%) & 0.081 $\pm$ 0.028 (18.3\%) \\  
    \bottomrule
    \end{tabular}
    }
    \label{tab:robustness_test}
    \end{table}
    
    In Table \ref{tab:robustness_test}, LAGMA shows not only the best performance but also the robust performance in various problem settings. 
    The fast learning of LAGMA is attributed to the latent goal-guided incentive, which generates accurate TD-target by utilizing values of semantically similar states projected to the same quantized vector.   
    Because LAGMA utilizes the value of semantically similar states rather than the specific states when learning Q-values, different yet semantically similar states tend to have similar Q-values, yielding generalizable policies. In this manner, LAGMA would enable further exploration, rather than solely enforcing exploitation of an identified state trajectory. Thus, even though the transition toward a goal-reaching trajectory is encouraged during training, the policy learned by LAGMA does not overfit to specific trajectories and exhibits robustness to unseen maps. 

    \subsection{Scalability test}
    LAGMA can be adopted to large-scale problems without any modifications. VQ-VAE takes a global state as an input to project them into quantized latent space. Thus, in large-scale problems, only the input size will differ from tasks with a small number of agents. In addition, many MARL tasks include high-dimensional global input size as presented in Table 1 in the manuscript. To assess the scalablity of LAGMA, we conduct additional experiments in 27m\_vs\_30m SMAC task, whose state dimension is 1170. Figure \ref{fig:large_scale} illustrates the performance of LAGMA. In Figure \ref{fig:large_scale}, LAGMA maintains efficient learning performance even when applied to large-scale problems, using identical hyperparameter settings as those for small-scale problems such as 5m\_vs\_6m.
     \begin{figure*}[!hbt]
        \centering
        \subfigure[Learning curve.]{\includegraphics[width=0.33\linewidth]{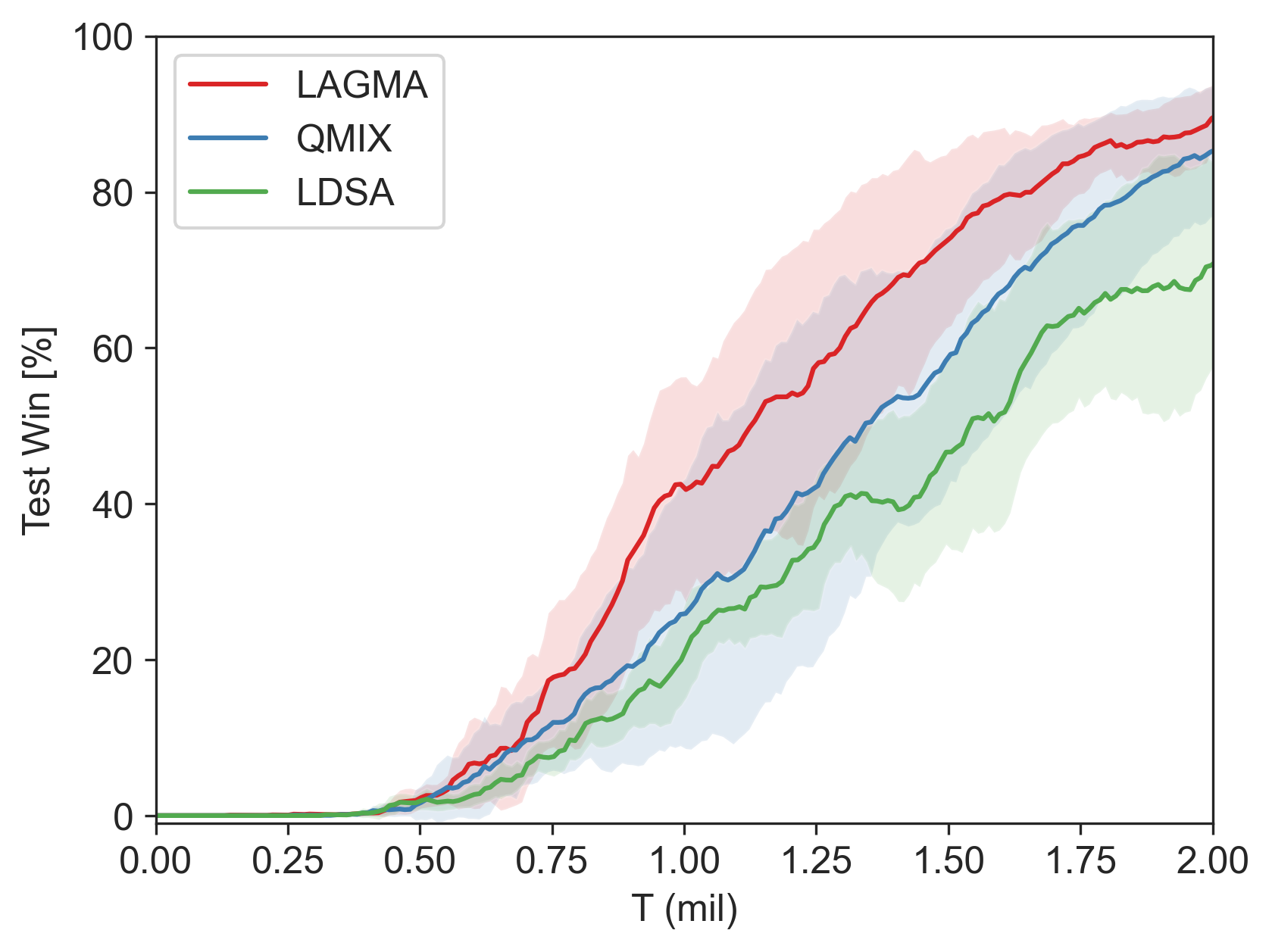}}    
        \caption{Performance on large-scale problem (27m\_vs\_30m SMAC).}
        \label{fig:large_scale}
        \vskip -0.15in
    \end{figure*}

    \subsection{Generalizability test to problems with diverse semantic goals} 
    To show the generalizability of LAGMA further, we conducted additional experiments on another benchmark such as SMACv2 \cite{ellis2024smacv2}, which includes diversity in initial positions and unit combinations within the identical task. Thus, from the perspective of latent space, SMACv2 tasks may encompass distinct multiple goals, even within the same task. For evaluation, we adopt the same hyperparameters as those for \texttt{3\_vs\_1WK} presented in Table 3 in the manuscript, except for $D=4$ and $n_{\textrm{freq}}^{cd}=40$.

    \begin{figure*}[!h]
        \centering
        \subfigure[\texttt{protoss\_5\_vs\_5}]{\includegraphics[width=0.32\linewidth]
        {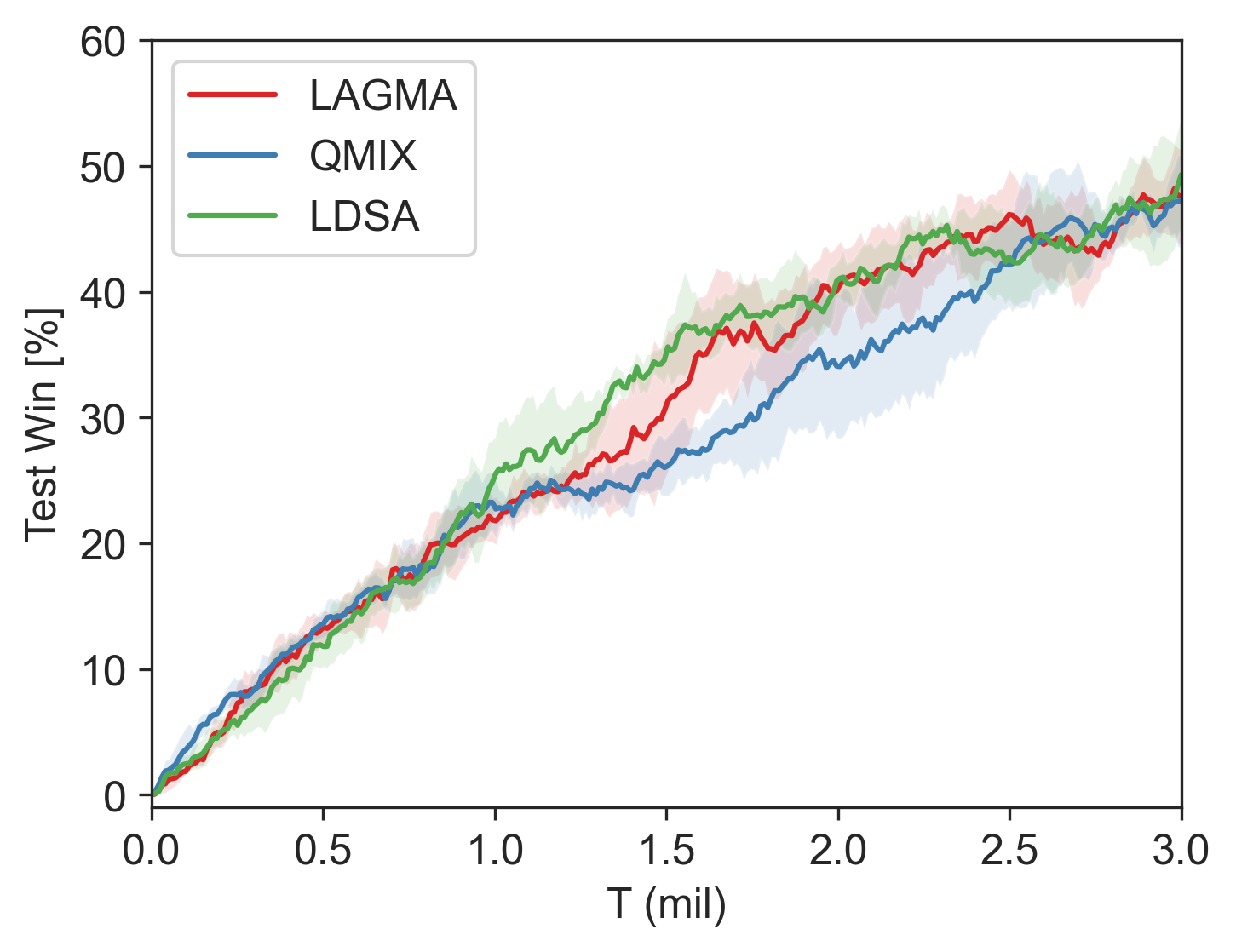}}    
        \subfigure[\texttt{protoss\_10\_vs\_11}]{\includegraphics[width=0.32\linewidth]{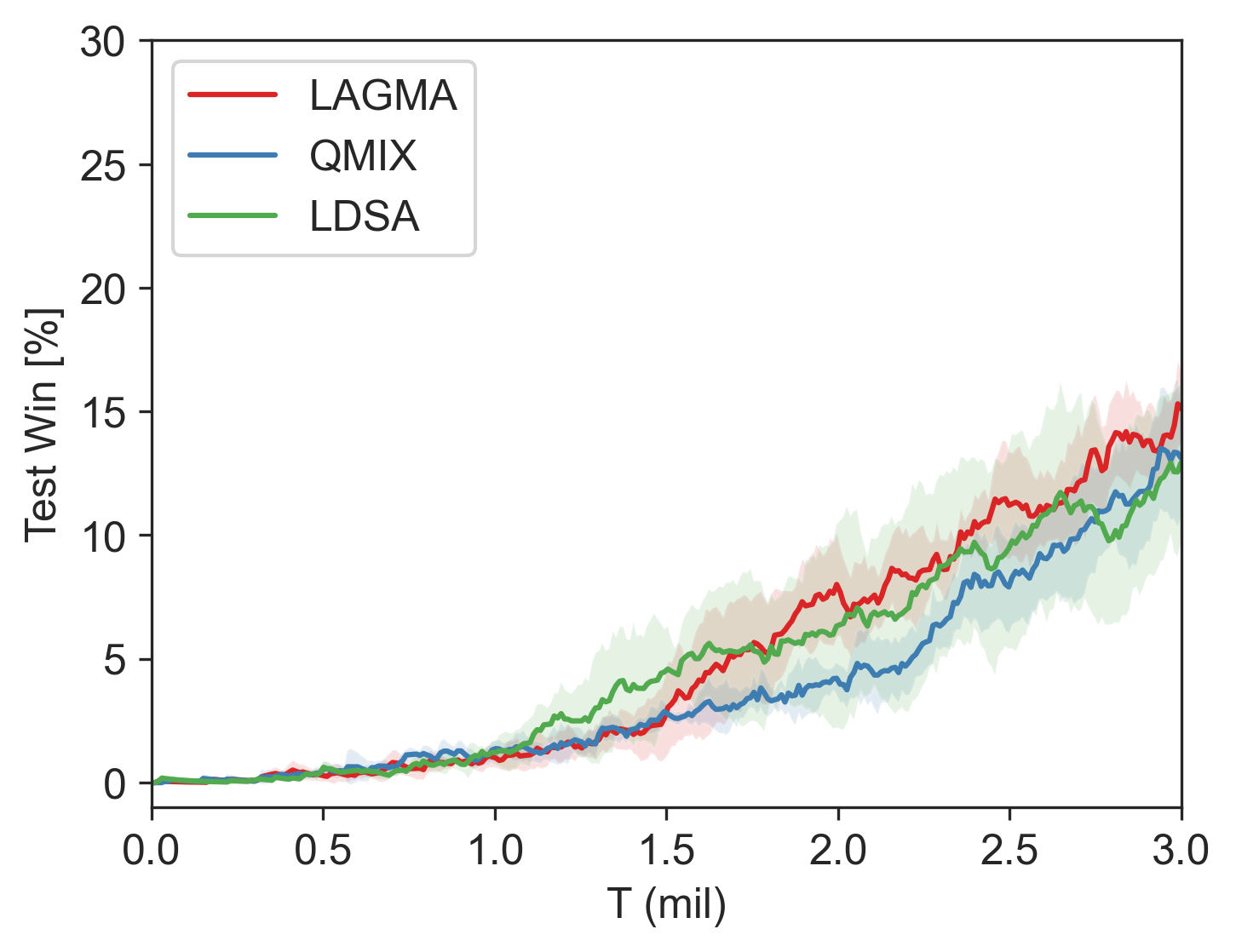}} 
        \caption{Performance evaluation on SMACv2.}
        \label{fig:smacv2}
    \end{figure*}

    In Figure \ref{fig:smacv2}, LAGMA shows comparable or better performance than baseline algorithms, but it does not exhibit distinctively strong performance, unlike other benchmark problems. We deem that this result stems from characteristics of the current LAGMA capturing reference trajectories towards a similar goal in the early training phase. 
   
    Multi-objective (or multiple goals) tasks may require a diverse reference trajectory generation. The current LAGMA only considers the return of a trajectory when storing reference trajectories in $\mathcal{D}_{seq}$. Thus, when trajectories toward different goals bifurcate from the same quantized vector, i.e., semantically similar states, they may not be captured by the current version of LAGMA algorithm if their return is relatively low compared to that of other reference trajectories already stored in $\mathcal{D}_{seq}$. Thus, LAGMA may not exhibit strong effectiveness in such tasks until various reference trajectories toward different goals are stored for a given quantized vector.   
    
    To improve, one may also consider the diversity of a trajectory when storing a reference trajectory in $\mathcal{D}_{seq}$. In addition, goal or strategy-dependent agent-wise execution would enhance coordination in such problem cases, but it may lead to delayed learning in easy tasks. The study regarding this trade-off would be an interesting direction for future research.

\newpage
    \section{Computational cost analysis}
    \label{computation_analysis}
    \subsection{Resource usage} 
    The introduction of an extended VQ codebook in LAGMA requires additional memory usage to an overall MARL framework. Memory usage depends on the codebook number ($n_c$), the number of a reference trajectory (index sequence) to save ($k$) in sequence buffer $\mathcal{D}_{seq}$, its batch time length ($T$), the total number of data saved for moving average computation ($m$), and data type. Memory usage of $\mathcal{D}_{\tau_{\chi_{t}}}$ and $\mathcal{D}_{VQ}$ are computed as follows.
    
    \begin{itemize}
        \item $\mathcal{D}_{\tau_{\chi_{t}}}: \texttt{byte(dtype)} \times n_c \times k \times T$
        \item $\mathcal{D}_{VQ}: \texttt{byte(dtype)} \times n_c \times m$        
    \end{itemize}
    
    For example, when $m=100$, $n_c=64, k=30$, $T=T_{max}=150$, i.e., the maximum timestep defined by the environment, and $\mathcal{D}_{\tau_{\chi_{t}}}$ and $\mathcal{D}_{VQ}$ use data type \texttt{int64} and \texttt{float32}, respectively, resource usages by introducing extended VQ codebook are computed as follows:
    
    \begin{itemize}
    \item $(\mathcal{D}_{\tau_{\chi_{t}}})_{\texttt{max}}$: $8(\texttt{int64})\times64\times30\times150=\texttt{2.19MiB}$    
    \item $\mathcal{D}_{VQ}$: $4(\texttt{float32})\times64\times100 = \texttt{25.6KiB}$
    \end{itemize}
    
    Here, $(\mathcal{D}_{\tau_{\chi_{t}}})_{\texttt{max}}$ value represents the possible maximum value and the actual value may vary based on the goal-reaching trajectory of each task. We can see that resource requirement due to the introduction of the extended codebook is marginal compared to that of the replay buffer and the GPU's memory capacity, such as $\texttt{24GiB}$ in GeForce RTX3090. Note that any of these memory usages do not depend on the dimension of states since only the index ($z$) of the quantized vector ($x_q$) of a sequence is stored in $\mathcal{D}_{\tau_{\chi_{t}}}$.

    \subsection{Training time analysis} 

    
    In LAGMA, we need to conduct an additional update for VQ-VAE and the extended codebook. Thus, the update frequency of VQ-VAE and the extended codebook would affect the overall training time. In the manuscript, we utilize the identical update interval $n_{\textrm{freq}}^{vq}=10$, indicating training once every 10 MARL training iterations for both VQ-VAE and codebook update. Table \ref{tab:training_time} represents the overall training time taken by various algorithms for diverse tasks. GeForce RTX3090 is used for \texttt{5m\_vs\_6m} and GeForce RTX4090 for \texttt{8m(sparse)} and \texttt{MMM2}. In the case of \texttt{8m(sparse)} task, the training time varies according to whether the learned model finds policy achieving a common goal. Thus, the training time of the successful case is presented for each algorithm.

    \begin{table}[!htbp]
    \centering
    \caption{Training time for each model in various SMAC maps (in hours).}
        \begin{tabular}{cccc}
        \toprule
        Model & \texttt{5m\_vs\_6m} (2M) & \texttt{8m(sparse)} (3M) & \texttt{MMM2} (3M) \\
        \midrule
        EMC   &  8.6            & 11.8            & 12.0     \\
        MASER & 12.7            & 13.4            & 20.5     \\
        LDSA  &  5.6            & 11.0            &  9.8     \\
        LAGMA & 10.5            & 12.6            & 17.7     \\
        \bottomrule
        \end{tabular}\label{tab:training_time}
    \end{table}

    
    Here, numbers in parenthesis represent the maximum training time ($T_{env}$) according to tasks. In Table \ref{tab:training_time}, we can see that training of LAGMA does not take much time compared to existing baseline algorithms. Therefore, we can conclude that the introduction of VQ-VAE and the extended codebook in LAGMA imposes an acceptable computational burden, with only marginal increases in resource requirements.


\end{document}